\documentclass[11pt]{article}

\usepackage{amssymb,latexsym}
\usepackage{amsmath}
\usepackage{amsfonts}
\usepackage{amsopn}
\usepackage{amsthm}
\usepackage{epsfig}
\usepackage{multirow}
\usepackage{url}
\usepackage[pdftex,bookmarks=true]{hyperref}
\evensidemargin\oddsidemargin \topmargin-15mm

\newcounter{NN}
\newtheorem{remark}[NN]{Remark}

\newtheorem{theorem}[NN]{Theorem}
\newtheorem{corollary}[NN]{Corollary}

\newtheorem{lemma}[NN]{Lemma}

\def\Z{\mathbb{Z}}
\def\R{\mathbb{R}}
\def\Z{\mathbb{Z}}

\def\mk{{\rm mKdV}}
\def\sg{{\rm sG}}
\def\pk{{\rm pKdV}}
\def\t{\Theta}

\begin{document}
\title{Involutivity of integrals of sine-Gordon, modified KdV and potential KdV maps}
\author{ Dinh T. Tran, P.H. van der Kamp, G.R.W. Quispel\\
Department of Mathematics, La Trobe University, Victoria, 3086, Australia\\
E-mail:dinhtran82@yahoo.com}
\date{\today }
\maketitle
\begin{abstract}
Closed form expressions in terms of multi-sums of products have been given in \cite{Tranclosedform, KRQ} of integrals of sine-Gordon,
modified Korteweg-de Vries and potential Korteweg-de Vries maps obtained as so-called $(p,-1)$-traveling wave reductions
of the corresponding partial difference equations. We prove the involutivity of these integrals with respect to
recently found symplectic structures for those maps. The proof is based on explicit formulae for the Poisson
brackets between multi-sums of products.
\end{abstract}
\section{Introduction}
Integrable systems boast a long and venerable history. The  history dates back to the 17th century with the work of Newton
on the two body problem. The notion of integrability was first introduced by Liouville in the 19th century in the
context of  finite dimensional
continuous Hamiltonian systems. Since then, it has been expanded to classes of nonlinear (partial) differential equations,
see for example \cite{ DynaSys4, Grammaticos2004Integrability-o}. More recently, there has been a shift of interest into systems with
discrete time, e.g. integrable ordinary difference  equations (or maps) and integrable partial difference (or lattice) equations.
Some of the first examples of discrete integrable systems appeared in \cite{Hirota, QCPN1984LinearIntegralEqs}. And a classification
of integrable lattice equations defined on a elementary square of the lattice has recently been obtained \cite{ABS},
based on the notion of multi-dimensional consistency. For maps there is the notion of complete or Liouville-Arnold integrability \cite{Bruschi1991,MaedaCompleteIntegrable,Ves},  analogues  to the same notion for continuous systems.
Briefly speaking, a mapping is said to be completely integrable if it has  a sufficient number of functionally independent integrals
that are in involution, that is, they Poison commute.

In this paper we study the involutivity of integrals of a certain class of integrable maps related to the fully discrete sine-Gordon,
modified Korteweg-de Vries (mKdV) and potential Korteweg-de Vries (pKdV) equations. These maps arise as travelling wave reductions
from the corresponding lattice equations. Such maps typically come in an infinite family of increasing dimension, and for this reason
it is not feasible to calculate Poisson brackets one by one and show that they all vanish. One  way to circumvent this problem
is to use the so-called Yang-Baxter structure, and that is the approach taken in \cite{CapelNijhoff1991,NCPquantummap}. This
approach was used to prove the involutivity of integrals for the so-called $(p,-p)$-reduction of the lattice pKdV equation. We refer
to \cite{QCPN,Kamp2009InitialValue} for the background on $(p,q)$-travelling wave reductions.
In this paper we study $(p,-1)$-reductions and we take a different approach. Starting from recently found symplectic structures \cite{Iatrou2003Higher-Dimension,RQs}, and recently obtained closed-form expressions in terms of multi-sums of products for integrals
of our family of sine-Gordon, mKdV and pKdV maps \cite{KRQ,Tranclosedform}, we proceed to prove involutivity of the integrals directly,
using explicit formulae for the Poisson brackets between multi-sums of products. These formulae will be proven by induction on the
number of variables, that is, on the dimension of the maps.

Recall, cf. \cite{Bruschi1991,Iatrou2003Higher-Dimension,Ves}, that a $2n$-dimensional discrete map $L:x\mapsto x^{'}$ is said to
be completely integrable if:
\begin{itemize}
\item there is a $2n\times 2n$  anti-symmetric non-degenerate  matrix $\Omega$  satisfying
the Jacobi identity
\[
\sum_{l}\left(\Omega_{li}\frac{\partial}{\partial x_l}\Omega_{jk}+\Omega_{lj}\frac{\partial}{\partial x_l}\Omega_{ki}
+\Omega_{lk}\frac{\partial}{\partial x_l}\Omega_{ij}\right)=0,
\]
such that $dL(x)\Omega(x)dL^T(x)=\Omega(x^{'})$, where $dL$ is the Jacobian of the map, $dL_{ij}:=\frac{\partial x^\prime_i}{\partial x_j}$.
\item there exist $n$ functionally independent integrals $I_1,I_2,\ldots, I_n$ satisfying
$\{I_r,I_s\}_x=0$ for all $1\leq r,s\leq n$, where the Poisson bracket is defined by
\begin{equation} \label{pb}
\{f,g\}_x = \nabla_x (f). \Omega. (\nabla_x (g))^T,
\end{equation}
with $\nabla_x=\frac{\partial}{\partial x_1}, \frac{\partial}{\partial x_2}, \ldots, \frac{\partial}{\partial x_{2n}}$.
Note that we will encounter several (related) Poisson brackets which are distinguished by the label $x$
 denoting the coordinates in which the bracket is expressed. Also, $\nabla_x$ will always have
the right number of components.
\end{itemize}

The families of ordinary difference sine-Gordon, mKdV and pKdV  equations
are given as follows, \cite{KRQ,Tranclosedform}
\begin{align}
\mbox{sine-Gordon}:& \ \alpha_1 (v_nv_{n+p+1}-v_{n+1}v_{n+p}) + \alpha_2
v_nv_{n+1}v_{n+p}v_{n+p+1} - \alpha_3 =0,\label{E:sG}\\
{\rm modified\ KdV}:&\  \beta_1(v_nv_{n+p}-v_{n+1}v_{n+p+1}) + \beta_2v_nv_{n+1}
-\beta_3v_{n+p}v_{n+p+1}=0,\label{E:mk}\\
{\rm potential\ KdV}:&\  (v_n-v_{n+p+1})(v_{n+1}-v_{n+p})-\gamma=0 \label{E:pk}.
\end{align}
These equations are obtained from the $(p,-1)$-traveling wave reductions of the corresponding partial difference equations of the form
\begin{equation}
\label{E:EqQuadgraph}
f(u_{l,m},u_{l+1,m},u_{l,m+1},u_{l+1,m+1})=0,
\end{equation}
where we have taken $v_n=u_{l,m}$ with $n=l+mp$, introducing the periodicity $u_{l,m}=u_{l+p,m-1}$, cf. \cite{QCPN,Kamp2009InitialValue}.

The  corresponding  $d=p+1$ dimensional maps derived from equations~(\ref{E:sG}), (\ref{E:mk}), (\ref{E:pk}) are $\R^{d}\rightarrow \R^{d}$,
\begin{equation}
\label{E:map}
(v_1,v_1,\ldots, v_d)\mapsto(v_2,v_3,\ldots, v_{d+1}),
\end{equation}
where $$v_{d+1}=v_1^{-1}\displaystyle{\frac{\alpha_1v_2v_{d}+\alpha_3}{\alpha_2v_2v_{d}+\alpha_1}},
\quad v_{d+1}=v_1\displaystyle{\frac{\beta_1v_d+\beta_2v_2}{\beta_1v_2+\beta_3v_d}},
\quad v_{d+1}=v_1-\frac{\gamma}{v_2-v_d},$$
respectively.
The integrals of sine-Gordon and mKdV maps can be expressed in terms of multi-sums of products,
which we call Theta:
\begin{equation} \label{D:Theta}
\Theta_{r,\epsilon}^{a,b}[f_i]: = \sum_{a \leq i_1 < i_2 < \cdots < i_r
\leq b} \prod_{j=1}^r (f_{i_j})^{(-1)^{j + \epsilon}},
\end{equation}
with $f_i=v_iv_{i+1}$. In \cite{KRQ} it was shown that $\lfloor d/2 \rfloor$ integrals of the sine-Gordon
map are given by
\begin{equation} \label{E:integral}
I^{\sg}_r = \alpha_1\left(\frac{v_d}{v_1}\t^{1,d-1}_{2r,1}+\frac{v_1}{v_d}\t^{1,d-1}_{2r,0}\right)
+\alpha_2\t^{1,d-1}_{2r+1,1}+\alpha_3\t^{1,d-1}_{2r+1,0},\quad 0\leq 2r < d-1
\end{equation}
and $\lfloor (d-1)/2 \rfloor$ integrals of the mKdV map are given by
\begin{equation} \label{E:mkdvinte}
I^{\mk}_r=\beta_1\left(v_1v_d\t^{1,d-1}_{2r-1,0}+\frac{1}{v_1v_d}\t^{1,d-1}_{2r-1,1}\right)
+\beta_2\t^{1,d-1}_{2r,1}+\beta_3\t^{1,d-1}_{2r,0}, \quad 0< 2r < d.
\end{equation}
In \cite{Tranclosedform} it was shown that $\lfloor (d-1)/2 \rfloor$ integrals of the pKdV map are given by
\begin{equation}
\label{E:pkinter}
I^{\pk}_r=\Psi^{2,d-2}_{r-1}+(v_d-v_2)\Psi^{2,d-3}_{r-1}+(v_{d-1}-v_1)\Psi^{3,d-2}_{r-1}
+\Psi^{3,d-3}_{r-2}+\left((v_{d-1}-v_1)(v_d-v_2)-\gamma\right)\Psi^{2,d-2}_{r},
\end{equation}
where $0 \leq r < \lfloor (d-1)/2\rfloor$ and
\begin{equation}
\label{D:Psi}
\Psi^{a,b}_r[f_i]= \left(\sum_{a\leq i_1,i_1+1<i_2,i_2+1,\ldots,<i_r\leq b}\prod_{j=1}^r f_{i_j}\right)\prod_{i=a}^{b+1}c_i,
\end{equation}
with $c_i=v_{i-1}-v_{i+1}$ and $f_i=1/(c_ic_{i+1})$.
In this paper we will prove that the integrals~(\ref{E:integral}),(\ref{E:mkdvinte}) and (\ref{E:pkinter}) are
in involution with respect to accompanying symplectic structures.

The paper is organized as follows. In section 2, we prove the involutivity of
integrals of  the sine-Gordon maps. Firstly,  we consider  the odd-dimensional maps. We introduce a transformation to reduce
 the dimension of the map by one and  we present a symplectic structure of the reduced map.
 Then we use properties of Theta with respect to the Poisson bracket associated to this symplectic structure. These
 properties are proven in Appendix A. To prove the involutivity
of the integrals, we write the Poisson bracket $\{I_r,I_s\}$ as a
 polynomial in $\alpha_1,\alpha_2,\alpha_3$ and prove that
all the coefficients of this polynomial vanish. Secondly, we consider the even-dimensional map. We provide a symplectic structure for it, and show that it relates to the symplectic structure for the odd dimensional map.
Therefore, many properties of Theta with respect to the new Poisson bracket can be obtained directly from the ones with respect
to the old Poisson bracket. The proof of involutivity is similar to the first case.

In section 3, we present relationships between symplectic structures  of the sine-Gordon and mKdV maps.
 We use these relationships  to derive  analogous properties of Theta
with respect to  the Poisson bracket of the mKdV maps. Involutivity of the integrals of the mKdV follows from these
properties.

In  section 4, we prove that the integrals of the pKdV map are in involution (with respect to
the approriate symplectic structures). We again  distinguish even and odd  dimensional maps and present a relationship
of symplectic structures between the two cases. For the even-dimensional map, the properties of multi-sums of products,
$\Psi$, with respect to the symplectic structure  are proved by induction in Appendix B. For the other case, the properties of $\Psi$
with respect to its symplectic structure are derived from the previous case. The involutiviy of integrals~(\ref{E:pkinter})
is proved by using these properties.


\section{Involutivity of sine-Gordon integrals\label{S:sGinvo}}
In this section, we distinguish two cases: the odd-dimensional and even-dimensional
 sine-Gordon maps. In \cite{KRQ} it is shown that for the even-dimensional map,  we have enough integrals
 for integrability. For the odd-dimensional map, we need to reduce the dimension of the map by one.
 We expand   the   Poisson bracket between two integrals $\{I_r,I_s\}$
as a quadratic  polynomial in the parameters $\alpha_1,\alpha_2,\alpha_3$ and prove
 the involutivity of
integrals (\ref{E:integral})  by showing its coefficients vanish.

\subsection{The case $d=2n+1$\label{SS:sGeven}}
 Using a reduction $f_i=v_iv_{i+1}$, we obtain a $2n$-dimensional map
\begin{equation}
\label{E:sGevenmap}
\sg: (f_1,f_2,\ldots, f_{2n})\mapsto (f_2,f_3,\ldots, f_{2n},\frac{f_2f_4\ldots f_{2n}(\alpha_1f_2f_4\ldots f_{2n}+\alpha_3f_3f_5\ldots f_{2n-1})}
{f_1f_3\ldots f_{2n-1}(\alpha_2f_2f_4\ldots f_{2n}+\alpha_1f_3f_5\ldots f_{2n-1})}.
\end{equation}
This map has $n$ integrals given by
\begin{equation}
\label{E:sGeveninte}
I^{\sg}_r=\alpha_1\left(\frac{f_2f_4\ldots f_{2n}}{f_1f_3\ldots f_{2n-1}}\t^{1,2n}_{2r,1}
+\frac{f_1f_3\ldots f_{2n-1}}{f_2f_4\ldots f_{2n}}\t^{1,2n}_{2r,0}\right)+\alpha_2\t^{1,2n}_{2r+1,1}+\alpha_3\t^{1,2n}_{2r+1,0},
\end{equation}
where the argument of $\t$ is $f_i$ and $0\leq r\leq n-1$.

A symplectic structure for the map~(\ref{E:sGevenmap}) is given by $\Omega^{\sg}_{2n}$, where
\begin{equation}
\label{E:sGevensym}
\Omega^{\sg}_{p}=\left(
\begin{matrix}
0& f_1f_2& f_1f_3&f_1f_4&\ldots &f_1f_{p-1}&f_1f_{p}\\
-f_2f_1& 0 & f_2f_3&f_2f_4&\ldots & f_2f_{p-2}&f_2f_{p}\\
-f_3f_1&-f_3f_2&0& f_3f_4&\ldots &f_3f_{p-1}&f_3f_{p}\\
\vdots& \vdots & \vdots & \vdots & \ldots & \vdots & \vdots\\
-f_{p}f_1 & -f_{p}f_2 & -f_{p}f_3 & -f_{p}f_4 & \ldots & -f_{p}f_{p-1}&0
\end{matrix}
\right).
\end{equation}
cf. \cite{Iatrou2003Higher-Dimension,RQs}.
One can verify that $d\sg \cdot \Omega^{\sg}_{2n} \cdot d\sg^T = \Omega^{\sg}_{2n} \circ \sg$.
Let   $g$ and $h$  be  functions differentiable with respect to the $f_i$'s. The symplectic structure
$\Omega^{\sg}_{p}$ defines the following Poisson bracket
\begin{eqnarray}
\{g,h\}_f&=&\nabla_f(g) \cdot \Omega^{\sg}_{p} \cdot (\nabla_f(h))^T \notag \\
&=&\sum_{i<j}f_if_j\left(\frac{\partial g}{\partial f_i}\frac{\partial h}{\partial f_j}
-\frac{\partial g}{\partial f_j}\frac{\partial h}{\partial f_i} \right).\label{E:sGevenPs}
\end{eqnarray}
We will prove that integrals (\ref{E:sGeveninte}) are in involution with respect to the symplectic
structure $\Omega^{\sg}_{2n}$, i.e $\{I^{\sg}_r, I^{\sg}_s\}_f=0$, for all $0\leq r,s\leq n-1$.
The proof is based on  the following explicit expressions for the Poisson bracket between Theta multi-sums, which
are proved in Appendix A.

\begin{lemma}
\label{L:PoiThetaeven}
Let $1\leq r,s\leq p$ and $\epsilon\in\{0,1\}$. We have
\begin{equation}\label{E:PoiTheta}
\{\Theta^{1,p}_{r,\epsilon},\Theta^{1,p}_{s,\epsilon}\}_f=\left\{\begin{array}{ll}
0&\  r,s \ \mbox{are both odd or both even},\\
\displaystyle{\sum_{i\geq 0}(-1)^i\Theta^{1,p}_{r+i,\epsilon}\Theta^{1,p}_{s-i,\epsilon}}
&\ r \  \mbox{even},\  s \ \mbox{odd and} \ r>s,\\
\displaystyle{\sum_{i\geq 1}(-1)^{i-1}\Theta^{1,p}_{r-i,\epsilon}\Theta^{1,p}_{s+i,\epsilon}}
&\ r\  \mbox{even},\  s \ \mbox{odd and} \ r<s.
\end{array}
\right. 
\end{equation}
Note that the right hand side of~(\ref{E:PoiTheta}) is a finite sum.
\end{lemma}

The next proposition provides the Poisson bracket between two Theta multi-sums with
different values of  $\epsilon$.

\begin{lemma}
\label{L:ThetaPoieven2}
Let $1\leq r,s\leq p$.
\begin{enumerate}
\item
If $r\equiv s \pmod 2$, we have
\begin{equation}
\label{E:ThetaPoieveniden1}
\{\Theta^{1,p}_{r,0},\Theta^{1,p}_{s,1}\}_f=\left\{
\begin{array}{lc}
\displaystyle{\sum_{i\geq 0}(-1)^i\t^{1,p}_{r-1-2\lfloor i/2\rfloor,i}\t^{1,p}_{s+1+2\lfloor i/2\rfloor,i+1}}
& \   r\leq s,\\
\displaystyle{\sum_{i\geq 0}}(-1)^i\t^{1,p}_{s-1-2\lfloor i/2\rfloor,i}\t^{1,p}_{r+1+2\lfloor i/2\rfloor,i+1}
& \  r>s.\\
\end{array}
\right. 
\end{equation}
\item
If $r\not\equiv s\pmod 2$, we have
\begin{equation}
\label{E:ThetaPoieveniden2}
\{\Theta^{1,p}_{r,0},\Theta^{1,p}_{s,1}\}_f=\left\{
\begin{array}{lc}
\displaystyle{\sum_{i\geq 0}} (-1)^i\Theta^{1,p}_{s+i,i+1}\Theta^{0,p}_{r-i,i}& \ r\ \mbox{odd},\ s\ \mbox{even},\\
\displaystyle{\sum_{i\geq 0}}(-1)^{i-1}\Theta^{1,p}_{s-i,i+1}\Theta^{1,p}_{r+i,i} & \ r \mbox{even},\  s\ \mbox{odd}.
\end{array}
\right. .
\end{equation}
\end{enumerate}
\end{lemma}

Using Lemma~{\ref{L:PoiThetaeven}} and Lemma~\ref{L:ThetaPoieven2}, we have the following  corollary.

\begin{corollary}\label{C:sGeven}
 Let $r$ and $s$  be both  even or both odd and let $\epsilon\in\{0,1\}$. Then
 \begin{align}
 \{\t^{1,p}_{r,0},\t^{1,p}_{s,1}\}_f+\{\t^{1,p}_{r,1},\t^{1,p}_{s,0}\}_f&=0, \label{E:sGevenco1}\\
\{\Theta^{1,p}_{r-1,\epsilon},\Theta^{1,p}_{s,\epsilon}\}_f+\{\Theta^{1,p}_{r,\epsilon},\Theta^{1,p}_{s-1,\epsilon}\}_f&=
\left\{\begin{array}{ll}
0, & \  r,s\ \mbox {even},\\
\Theta^{1,p}_{r-1,\epsilon}\Theta^{1,p}_{s,\epsilon}-\Theta^{1,p}_{s-1,\epsilon}\Theta^{1,p}_{r,\epsilon},& \  r,s \ \mbox{ odd},
\end{array}
\right.
\label{E:sGevenco2}\\
\{\Theta^{1,p}_{r-1,\epsilon\pm 1},\Theta^{1,p}_{s,\epsilon}\}_f+\{\Theta^{1,p}_{r,\epsilon},\Theta^{1,p}_{s-1,\epsilon\pm 1}\}_f&=
\left\{\begin{array}{ll}
0, & \   r,s\ \mbox {even},\\
\Theta^{1,p}_{s-1,\epsilon\pm 1}\Theta^{1,p}_{r,\epsilon}-\Theta^{1,p}_{s,\epsilon}\Theta^{1,p}_{r-1,\epsilon\pm 1},& \  r,s \ \mbox{odd}.
\end{array}
\right. 
\label{E:sGevenco3}
 \end{align}
 \end{corollary}

\begin{theorem}
\label{T:sgeven}
Let $0\leq r,s\leq n-1$.  Let $I^{\sg}_r,I^{\sg}_s$ be  given by formula (\ref{E:sGeveninte}). Then
\[
\{I^{\sg}_r,I^{\sg}_s\}_f=0.
\]
\end{theorem}
\begin{proof}
First of all, we  denote
\[
F=\frac{f_1f_3\ldots f_{2n-1}}{f_2f_4\ldots f_{2n}}.
\]
For any $g(f_1,f_2,\ldots,f_{2n})$  we find
$\{F^{\pm 1},g\}_f=\pm F^{\pm 1} E_fg$, where
\begin{equation}
\label{E:Ef_operator}
E_f=\sum_{i\geq 1}f_i\frac{\partial}{\partial f_i},
\end{equation}
which scales any homogeneous expression by its total degree.
Every term in the multi-sum has total degree $0$ if $r$ is even and degree $(-1)^{\epsilon+1}$
if $r$ is odd, hence
\begin{equation}
\label{PFT}
\{F^{\pm 1},\t^{1,p}_{r,\epsilon}\}_f=\left\{
\begin{array}{ll}
0&\ \mbox {if} \ r \ \mbox{even},\\
\mp (-1)^{\epsilon} F^{\pm 1} \t^{1,p}_{r,\epsilon}&\  \mbox{if}\ r\ \mbox{odd}.
\end{array}
\right. 
\end{equation}

Now we expand $\{I^{\sg}_r,I^{\sg}_s\}_f$ in terms of polynomials in $\alpha_1,\alpha_2,\alpha_3$ as follows
\[
\{I^{\sg}_r,I^{\sg}_s\}_f=\alpha_1^2A_1+\alpha_2^2A_2+\alpha_3^2 A_3+\alpha_1\alpha_2A_{12}+\alpha_1\alpha_3A_{13}+\alpha_2\alpha_3A_{23},
\]
where
\begin{align*}
A_1&=\{
F^{-1}\t^{1,2n}_{2r,1}+F\t^{1,2n}_{2r,0},F^{-1}\t^{1,2n}_{2s,1}+F\t^{1,2n}_{2s,0}\}_f,\\
A_2&=\{\t^{1,2n}_{2r+1,1},\t^{1,2n}_{2s+1,1}\}_f,\\
A_3&=\{\t^{1,2n}_{2r+1,0},\t^{1,2n}_{2s+1,0}\}_f,\\
A_{12}&=\{F^{-1}\t^{1,2n}_{2r,1}+F\t^{1,2n}_{2r,0},\t^{1,2n}_{2s+1,1}\}_f +
\{\t^{1,2n}_{2r+1,1},F^{-1}\t^{1,2n}_{2s,1}+F\t^{1,2n}_{2s,0}\}_f,\\
A_{13}&=\{F^{-1}\t^{1,2n}_{2r,1}+F\t^{1,2n}_{2r,0},\t^{1,2n}_{2s+1,0}\}_f
+\{\t^{1,2n}_{2r+1,0},F^{-1}\t^{1,2n}_{2s,1}+F\t^{1,2n}_{2s,0}\}_f,\\
A_{23}&=\{\t^{1,2n}_{2r+1,1},\t^{1,2n}_{2s+1,0}\}_f+\{\t^{1,2n}_{2r+1,0},\t^{1,2n}_{2s+1,1}\}_f.
\end{align*}
We prove that all these  coefficients equal $0$.
Using  Lemma  \ref{L:PoiThetaeven} and Corollary~\ref{C:sGeven}, we have $A_2=A_3=A_{23}=0$. We now expand
$A_1, A_{12}$ and $A_{13}$,
we obtain
\begin{align*}
A_1&=F^{-2}\{\t^{1,2n}_{2r,1},\t^{1,2n}_{2s,1}\}+\{\t^{1,2n}_{2r,1},\t^{1,2n}_{2s,0}\}+\{\t^{1,2n}_{2r,0},\t^{1,2n}_{2s,1}\}
+F^2\{\t^{1,2n}_{2r,0},\t^{1,2n}_{2s,0}\} \\
& \quad +\t^{1,2n}_{2r,1}\t^{1,2n}_{2s,1}\{F^{-1},F^{-1}\}+\t^{1,2n}_{2r,0}\t^{1,2n}_{2s,1}\{F,F^{-1}\}+
\t^{1,2n}_{2r,1}\t^{1,2n}_{2s,0}\{F^{-1},F\}+\t^{1,2n}_{2r,0}\t^{1,2n}_{2s,0}\{F,F\}\\
& \quad + F^{-1}\left(\t^{1,2n}_{2s,1}\{\t^{1,2n}_{2r,1},F^{-1}\}_f
+\t^{1,2n}_{2r,1}\{F^{-1},\t^{1,2n}_{2s,1}\}_f
+\t^{1,2n}_{2s,0}\{\t^{1,2n}_{2r,1},F\}_f+\t^{1,2n}_{2r,0}\{F,\t^{1,2n}_{2s,1}\}_f
\right)\\
& \quad +F\left(
\t^{1,2n}_{2r,1}\{F^{-1},\t^{1,2n}_{2s,0}\}_f+\t^{1,2n}_{2s,1}\{\t^{1,2n}_{2r,0},F^{-1}\}_f
 +\t^{1,2n}_{2r,0}\{F,\t^{1,2n}_{2s,0}\}_f+\t^{1,2n}_{2s,0}\{\t^{1,2n}_{2r,0},F\}_f
\right)\\
&=0,
\end{align*}
where the second and third terms cancel each other, due to (\ref{E:sGevenco1}), and all other terms vanish
according to equations (\ref{E:PoiTheta}),  and (\ref{PFT}).

We also get
\begin{align*}
A_{12}&=\ F^{-1}
\left( \{\t^{1,2n}_{2r,1},\t^{1,2n}_{2s+1,1}\}_f+\{\t^{1,2n}_{2r+1,1},\t^{1,2n}_{2s,1}\}\right)_f + F
\left(\{\t^{1,2n}_{2r,0},\t^{1,2n}_{2s+1,1}\}_f+\{\t^{1,2n}_{2r+1,1},\t^{1,2n}_{2s,0}\}_f\right)
\\
&\quad + \t^{1,2n}_{2r,1}\{F^{-1},\t^{1,2n}_{2s+1,1}\}_f
+\t^{1,2n}_{2r,0}\{F,\t^{1,2n}_{2s+1,1}\}_f
+ \t^{1,2n}_{2s,1}\{\t^{1,2n}_{2r+1,1},F^{-1}\}_f
+\t^{1,2n}_{2s,0}\{\t^{1,2n}_{2r+1,1},F\}_f
\\
&=\ F^{-1}
\left(\t^{1,2n}_{2r,1}\t^{1,2n}_{2s+1,1}-\t^{1,2n}_{2r+1,1}\t^{1,2n}_{2s,1}\right)
+F\left(\t^{1,2n}_{2s,0}\t^{1,2n}_{2r+1,1}-\t^{1,2n}_{2s+1,1}\t^{1,2n}_{2r,0}\right)
\\
&\quad - F^{-1}\t^{1,2n}_{2r,1}
\t^{1,2n}_{2s+1,1}
+ F\t^{1,2n}_{2r,0}
\t^{1,2n}_{2s+1,1}
 +F^{-1}\t^{1,2n}_{2s,1}
 \t^{1,2n}_{2r+1,1}
-F\t^{1,2n}_{2s,0}
\t^{1,2n}_{2r+1,1}
\\
&=0,
\end{align*}
where we have used  (\ref{E:sGevenco2}), (\ref{E:sGevenco3}), and (\ref{PFT}).
Similarly we get $A_{13}=0$. Hence, we have $\{I^{\sg}_r,I^{\sg}_s\}_f=0$.
\end{proof}
\subsection{The case  $d=2n$ \label{SS:sGodd}}
In this section, we consider a $2n$-dimensional map
\begin{equation}
 \label{E:sGodd}
 \widetilde{\sg}: (v_1,v_2,\cdots, v_{2n})\mapsto (v_2, v_3,\ldots,v_{2n}, v_1^{-1}\frac{\alpha_1v_2v_{2n}+\alpha_3}{\alpha_2v_2v_{2n}+\alpha_1}).
 \end{equation}
This map has $n$ integrals given by
\begin{equation}
\label{E:sGoddinte}
I^{\widetilde{\sg}}_{r}=\alpha_1\left(\frac{v_{2n}}{v_1}\Theta^{1,2n-1}_{2r,1}+\frac{v_1}{v_{2n}}\Theta^{1,2n-1}_{2r,0}\right)
+\alpha_2\Theta^{1,2n-1}_{2r+1,1}+\alpha_3\Theta^{1,2n-1}_{2r+1,0},
\end{equation}
where $0\leq r\leq n-1$ and the argument of Theta is $f_i=v_iv_{i+1}$. The sine-Gordon map~(\ref{E:sGodd}) has a symplectic structure $\Omega_{2n}^{\widetilde{\sg}}$, where
\begin{equation}
\label{E:sGoddsym}
\Omega_{p}^{\widetilde{\sg}}=\left(
 \begin{matrix}
 0 & v_1v_2&0&v_1v_4&\ldots &0&v_1v_{p}\\
 -v_2v_1&0&v_2v_3&0&\ldots&v_2v_{p-1}&0\\
0&-v_3v_2&0&v_3v_4&\ldots&0&v_3v_{p}\\
 \vdots&\vdots&\vdots&\vdots&\vdots&\vdots&\vdots\\
 -v_{p}v_1&0&-v_{p}v_3&0&\ldots&-v_{p}v_{p-1}&0
  \end{matrix}
 \right),
\end{equation}
cf. \cite{Iatrou2003Higher-Dimension,RQs}.
The Poisson bracket $\frac{1}{2}\nabla_v (g) \Omega_{p}^{\widetilde{\sg}} \left(\nabla_v (h)\right)^{T}$ is denoted $\{g,h\}_v$.
Before we prove that the integrals (\ref{E:sGoddinte}) are in involution with respect to this bracket, we first
establish the following Poisson brackets between Theta multi-sums:
\[
\{\Theta^{1,p}_{r,\epsilon},\Theta^{1,p}_{s,\delta}\}_v = \{\Theta^{1,p}_{r,\epsilon},\Theta^{1,p}_{s,\delta}\}_f |_{f_i=v_iv_{i+1}},
\]
where the right-hand-side is given by Propositions \ref{L:PoiThetaeven} and \ref{L:ThetaPoieven2}.
This follows as a corollary from the next Lemma. Consider the map $\R^p \rightarrow \R^{p-1}$,
\[
G_p:(v_1,v_2,\ldots,v_p ) \mapsto (v_1v_2,v_2v_3,\ldots, v_{p-1}v_p).
\]
\begin{lemma} \label{PBR1}
With $g,h$ differentiable functions on $\R^{p-1}$ we have
\[
\{ g \circ G_p , h\circ G_p\}_v =  \{ g , h \}_{f=G_p(v)},
\]
i.e. $G_p$ is a  Poisson map.
\end{lemma}
\begin{proof}
The $(p-1) \times p$ Jacobian of the map $G_p$ is
\[
dG_p=\left( \begin{array}{ccccc}
v_2 & v_1 & 0 & \cdots & 0 \\
0 & v_3 & v_2 & \cdots & 0 \\
\vdots & & & & 0 \\
0 & 0 & \cdots & v_p & v_{p-1}
\end{array}
\right).
\]
By direct calculation, we have
\begin{equation}
 \label{E:sGrelation}
 dG_p \cdot \Omega^{\widetilde{\sg}}_p \cdot (dG_p)^{T} = 2 \Omega^{\sg}_{p-1}|_{f=G_p(v)}.
 \end{equation}
Applying $\nabla$ to $(g\circ G_p) (v) = g (f) |_{f=G_p(v)}$ (and omitting some arguments) we find
\[
\nabla_v (g \circ G_p) = \nabla_f (g) dG_p |_{f=G_p(v)}.
\]
Hence, we have
 \begin{eqnarray*}
 \{ g \circ G_p, h \circ G_p \}_v &=&\frac{1}{2} \nabla_v (g \circ G_p) \Omega^{\widetilde{\sg}}_p \nabla_v (g \circ G_p)^{T} \\
 &=& \frac{1}{2}\nabla_f (g) dG_p |_{f=G_p(v)} \Omega^{\widetilde{\sg}}_p (\nabla_f (h) dG_p)^T |_{f=G_p(v)} \\
 &=& \nabla_f (g)  \Omega^{\sg}_p (\nabla_f (h))^T |_{f=G_p(v)} \\
 &=&  \{g,h\}_{f=G_p(v)}
 \end{eqnarray*}
 \end{proof}

Now we will prove the involutivity of the integrals (\ref{E:sGoddinte}) of the sine-Gordon map (\ref{E:sGodd}).
\begin{theorem}
\label{T:sGodd}
Let  $I^{\widetilde{\sg}}_r$ and $I^{\widetilde{\sg}}_s$, with $0\leq r,s\leq n-1$, be given by the formula
(\ref{E:sGoddinte}). Then we have
\[
\{I^{\widetilde{\sg}}_r,I^{\widetilde{\sg}}_s\}_v=0.
\]
\end{theorem}
\begin{proof}
 With $V=v_1/v_{2n}$ we have
\begin{equation*}
\{V^{\pm 1},\t^{1,p}_{r,\epsilon}\}_v=V^{\pm 1}E_v\t^{1,p}_{r,\epsilon}= \left\{
\begin{array}{ll}
0&\ \mbox {if} \ r \ \mbox{even},\\
\mp (-1)^{\epsilon} V^{\pm} \t^{1,p}_{r,\epsilon}&\  \mbox{if}\ r\ \mbox{odd}.
\end{array}
\right.
\end{equation*}

The Poisson bracket between 2 integrals is expanded as
\[
\{I^{\widetilde{\sg}}_r,I^{\widetilde{\sg}}_s\}_v=\alpha_1^2B_1+\alpha_2 ^2 B_2+\alpha_3 ^2 B_3 +\alpha_1\alpha_2 B_{12}+ \alpha_1\alpha_3 B_{13}+\alpha_2\alpha_3 B_{23},
\]
where the coefficients $B_I$ are similar to the $A_I$ given in section \ref{SS:sGeven}, replacing $F$ by $V$ and $2n$ by $2n-1$.
The rules for simplification are also similar.
Therefore, $\{I^{\widetilde{\sg}}_r,I^{\widetilde{\sg}}_s\}_v=0$.
\end{proof}
\section{Involutivity of  mKdV integrals\label{A:mkdv}}
We consider the $d$-dimensional mKdV map
\begin{equation}
\label{E:mkdvmap}
(v_1,v_1,\ldots, v_d)\mapsto\left(v_2,v_3,\ldots, v_d, v_1 \displaystyle{\frac{\beta_1 v_{d}+\beta_2 v_2}{\beta_1 v_2+\beta_3v_{d}}}\right).
\end{equation}
As shown in \cite{KRQ}, this map has $\lfloor (d-1)/2 \rfloor$ integrals given by the formula (\ref{E:mkdvinte}) with $0<2r<d$.
If $d=2n+1$, the map (\ref{E:mkdvmap})  reduces to a $2n$-dimensional map with exactly
$n$ integrals via a reduction $z_i=v_{i+1}/v_i$. For the other case, where $d=2n+2$, the map (\ref{E:mkdvmap})
reduces to a $2n$-dimensional map with exactly $n$ integrals via the reduction $z_i=v_{i+2}/v_{i}$.
We will show that the integrals of these reduced maps are in involution.
In each case, we present a relationship between the relevant symplectic structures and the symplectic structures of the
sine-Gordon map in the even case (\ref{E:sGevensym}). This relation can be used to derive properties of Theta with new
symplectic structures.

\subsection{The case $d=2n+1$\label{AS:mkdveven}}
Using the reduction $z_i=v_{i+1}/v_i$, we obtain the  map
\begin{equation}
\label{E:mkdvevenmap}
\mk:
(z_1,z_2,\ldots, z_{2n})\mapsto (z_2,z_3,\ldots,z_{2n},\frac{1}{z_1z_2\ldots z_{2n}}.\
\frac{\beta_1 z_2z_3\ldots z_{2n}+\beta_2}{\beta_1+\beta_3z_2z_3\ldots z_{2n}}).
\end{equation}

The integrals of this map are given by
\begin{equation}
\label{E:mkdveveninte}
I^{\mk}_r=\beta_1\left(z_1z_2\ldots z_{2n}\t^{1,2n}_{2r-1,0}+\frac{1}{z_1z_2\ldots z_{2n}}\t^{1,2n}_{2r-1,1}\right)
+ \beta_2\t^{1,2n}_{2r,1}+\beta_3\t^{1,2n}_{2r,0},
\end{equation}
where arguments for Theta are $f_i=z_1^2z_2^2\ldots z_{i-1}^2z_{i}$. Here we have used an 'inverse reduction',
$v_{i}=v_1z_1z_2\cdots z_{i-1}$ to express $f_i=v_iv_{i+1}$ in terms of the $z_j$ and we omitted the $v_1$ dependence
as both the integral and the map do not depend on it.

We obtain a symplectic structure $\Omega^{\mk}_{2n}$ for the map~(\ref{E:mkdvevenmap}), where
\begin{equation}
\label{E:mkdvevensymG}
\Omega^{\mk}_ p:=\left(
\begin{matrix}
0& z_1z_2&-z_1z_3&z_1z_4&\ldots   &(-1)^p z_1z_{p}\\
-z_2z_1&0& z_2z_3& z_2z_4&\ldots &(-1)^{p-1}z_2z_{p}\\
z_3z_1&-z_3z_2&0&z_3z_4&\ldots   &(-1)^{p-2}z_3z_{p}\\
\vdots&\vdots&\vdots&\vdots&\cdots&\vdots\\
(-1)^{p-1}z_{p}z_1&(-1)^{p}z_{p}z_2&(-1)^{p-1}z_{p}z_3& (-1)^{p}z_{p}z_4&\ldots &0
\end{matrix}
\right),
\end{equation}
cf. \cite{Iatrou2003Higher-Dimension,RQs}. This gives us a Poisson bracket $\{g,h\}_z =\nabla_z (g) \Omega_{2n}^{\mk} \left(\nabla_z (h)\right)^{T}$.
As before we can express the $z$-Poisson brackets between Theta multi-sums in terms of the
corresponding $f$-Poisson brackets. Consider the map
\[
M_p: (z_1,z_2,\ldots, z_{p}) \mapsto (z_1,z_1^2z_2,\ldots,z_1^2z_2^2\cdots z_{p-1}^2z_p).
\]
We have the following result.
\begin{lemma} \label{PBR2}
With $g,h$ differentiable functions on $\R^{p}$ we have
\[
\{ g \circ M_p , h\circ M_p\}_z = \{ g , h \}_{f=M_p(z)},
\]
i.e. $M_p$ is a Poisson map.
\end{lemma}
\begin{proof}
The $p \times p$ Jacobian of the map $M_p$ is
\[
dM_p=\left\{ \begin{array}{ll}
0 & i<j, \\
\prod_{k=1}^{i-1} z_k^2 & i=j, \\
2z_iz_j^{-1}\prod_{k=1}^{i-1} z_k^2 & i>j. \\
\end{array}
\right.
\]
and a calculation shows
\begin{equation}
\label{mkdvsymrela}
dM_p.\Omega^{\mk}_{p}.dM_p^{T}=\Omega^{\sg}_{p}\mid_{f=M_p(z)}.
\end{equation}
The argument is finished along the lines of the proof for Lemma \ref{PBR1}.
\end{proof}

We are ready now to prove the following theorem
\begin{theorem}
Let $I^{\mk}_r$ and $I^{\mk}_s$ be given by the formula~(\ref{E:mkdveveninte}) with $1\leq r,s\leq n$.
 Then we have
\begin{equation}
\label{E:MkdvInvo1}
\{I^{\mk}_r,I^{\mk}_s\}_z=0.
\end{equation}
\end{theorem}

\begin{proof}
With $Z=(z_1z_2\ldots z_{2n})^{-1}$ we have $F^{\pm 1} \circ M_{2n} = Z^{\pm 1}$. Thus, Lemma \ref{PBR2} implies
\begin{equation} \label{ZT}
\{Z^{\pm 1},\Theta^{1,2n}_{r,\epsilon}\}_z=\left\{
\begin{array}{ll}
0&\ \mbox {if} \ r \ \mbox{even},\\
\mp (-1)^{\epsilon} Z^{\pm 1} \t^{1,2n}_{r,\epsilon}&\  \mbox{if}\ r\ \mbox{odd}.
\end{array}
\right. 
\end{equation}

Writing the left hand side of equation \eqref{E:MkdvInvo1} as
\begin{equation}
\{I^{\mk}_r,I^{\mk}_s\}_z=\beta_1^2P_1+\beta_2^2P_2+\beta_3^2P_3+
\beta_1\beta_2P_{12}+\beta_1\beta_3P_{13}+\beta_2\beta_3P_{23},
\end{equation}
yields coefficients $P_I$ similar to the $A_I$ given in section \ref{SS:sGeven}, replacing $F$ by $Z$, $2r$ by $2r-1$,
and $2s$ by $2s-1$. Now that we know the brackets between $Z$, $Z^{-1}$, and the $\Theta^{1,2n}_{2s,1}$, we can expand
the coefficient and show they vanish.

As before, the coefficients $P_{2}$, $P_3$, and $P_{23}$ are the easy ones. For $P_1$ we get, using equation (\ref{ZZ} and
Lemma \ref{PBR2} in conjunction with equations $(\ref{E:sGevenco1})$ and (\ref{E:PoiTheta}),
\begin{align*}
P_1
&=
Z\left(\Theta^{1,2n}_{2s-1,0}\{\Theta^{1,2n}_{2r-1,0},Z\}_z+\Theta^{1,2n}_{2r-1,0}\{Z,\Theta^{1,2n}_{2s-1,0}\}_z
+\Theta^{1,2n}_{2s-1,1}\{\Theta^{1,2n}_{2r-1,0},Z^{-1}\}_z+\Theta^{1,2n}_{2r-1,1}\{Z^{-1},\Theta^{1,2n}_{2s-1,0}\}_z
\right)
\\
&\quad +
Z^{-1}\left(\Theta^{1,2n}_{2s-1,1}\{\Theta^{1,2n}_{2r-1,1},Z\}_z+\Theta^{1,2n}_{2r-1,1}\{Z^{-1},\Theta^{1,2n}_{2s-1,1}\}_z
+\Theta^{1,2n}_{2r-1,0}\{Z,\Theta^{1,2n}_{2s-1,1}\}_z+\Theta^{1,2n}_{2s-1,0}\{\Theta^{1,2n}_{2r-1,1},Z\}_z
\right)
\\
&\quad +
\left(\{\Theta^{1,2n}_{2r-1,0},\Theta^{1,2n}_{2s-1,1}\}_z+\{\Theta^{1,2n}_{2r-1,1},\Theta^{1,2n}_{2s-1,0}\}_z\right)
\\
&=Z^2\left(-\Theta^{1,2n}_{2s-1,0}\Theta^{1,2n}_{2r-1,0}+\Theta^{1,2n}_{2r-1,0}(-1)^{1}\Theta^{1,2n}_{2s-1,0}\right)+
Z^{-2}\left(-\Theta^{1,2n}_{2s-1,1}\Theta^{1,2n}_{2r-1,1}+\Theta^{1,2n}_{2r-1,1}\Theta^{1,2n}_{2s-1,1}\right)
\\
&\quad +
\left(\{\Theta^{1,2n}_{2r-1,0},\Theta^{1,2n}_{2s-1,1}\}_z+\{\Theta^{1,2n}_{2r-1,1},\Theta^{1,2n}_{2s-1,0}\}_z\right)
\\
&\quad +
\Theta^{1,2n}_{2s-1,1}\Theta^{1,2n}_{2r-1,0}-\Theta^{1,2n}_{2r-1,0}\Theta^{1,2n}_{2s-1,1}+
\Theta^{1,2n}_{2s-1,0}\Theta^{1,2n}_{2r-1,1}-\Theta^{1,2n}_{2r-1,1}\Theta^{1,2n}_{2s-1,0}
\\
&=0,
\end{align*}
where we have used (\ref{ZT}) and and~\eqref{E:sGevenco1}.

Expanding $P_{12}$ yields
\begin{align*}
P_{12}&=Z\left(\{\Theta^{1,2n}_{2r-1,0},\Theta^{1,2n}_{2s,1}\}_z+\{\Theta^{1,2n}_{2r,1},\Theta^{1,2n}_{2s-1,0}\}_z\right)+
Z^{-1}\left(\{\Theta^{1,2n}_{2r-1,1},\Theta^{1,2n}_{2s,1}\}_z+\{\Theta^{1,2n}_{2r,1},\Theta^{1,2n}_{2s-1,1}\}_z\right)
\\
&\quad +
\Theta^{1,2n}_{2r-1,0}\{Z,\Theta^{1,2n}_{2s,1}\}_z+\Theta^{1,2n}_{2s-1,0}\{\Theta^{1,2n}_{2r,1},Z\}_z+
\Theta^{1,2n}_{2r-1,1}\{Z^{-1},\Theta^{1,2n}_{2s,1}\}+\Theta^{1,2n}_{2s-1,1}\{\Theta^{1,2n}_{2r,1},Z^{-1}\}_z\\
&=0,
\end{align*}
where the first  and the second terms, as well as  the third and the fourth  terms cancel each other,
 and the last four terms are equal to zero.

Similarly, we get $P_{13}=0$.
\end{proof}

\subsection{The case $d=2n+2$\label{AS:mkdvodd}}
Now using  a reduction $w_i=v_{i+2}/v_i$, we obtain the map
\begin{equation}
\label{E:mkdvoddmap}
\widetilde{\mk}: \ (w_1,w_2,\ldots, w_{2n})\mapsto \left(w_2,w_3,\ldots,w_{2n},\frac{1}{w_1w_3\ldots w_{2n-1}}. \
\frac{\beta_1 w_2w_4\ldots w_{2n}+\beta_2}{\beta_1+\beta_3w_2w_4\ldots w_{2n}}\right).
\end{equation}
Integrals of this map are given by
\begin{equation}
\label{E:mkdvoddinte}
I^{\widetilde{\mk}}_r=\alpha_1\left(w_2w_4\ldots w_{2n}\t^{1,2n+1}_{2r-1,0}+\frac{1}{w_2w_4\ldots w_{2n}}\t^{1,2n+1}_{2r-1,1}\right)
+ \alpha_2\t^{1,2n+1}_{2r,1}+\alpha_3\t^{1,2n+1}_{2r,0},
\end{equation}
where $\t=\t[e_i]$ with $e_i=f_{i-1}$, with $f_0=1$ and $f_{i}=w_1w_2\ldots w_{i}$ ($i>0$).
Note, we have changed notation in order to relate the next Poisson bracket to the bracket $\{,\}_f$.
The argument of $\Theta$ is $v_iv_{i+1}=:e_i$. In the 'inverse reduction', we have
\[
v_n=\left\{\begin{array}{ll}
v_1 \prod_{j=1}^i w_{2j-1} & n=2i+1, \\
v_2 \prod_{j=1}^{i-1} w_{2j} & n=2i.
\end{array}\right.
\]
 Therefore (similar to the case $d=2n+1$) both the reduced map as well as  the reduced integrals depend on the variables $w_i$.
Using \eqref{E:recur2}, we obtain
\begin{eqnarray*}
\t^{1,2n+1}_{s,\epsilon}[e_i]&=&\t^{2,2n+1}_{s,\epsilon}[e_i]+\t^{2,2n+1}_{s-1,\epsilon+1}[e_i]\\
&=&\t^{1,2n}_{s,\epsilon}[f_i]+\t^{1,2n}_{s-1,\epsilon+1}[f_i].
\end{eqnarray*}
Let
\[
K_p: (w_1,w_2,\ldots, w_{p}) \mapsto (w_1,w_1w_2,\ldots,w_1w_2\cdots w_p).
\]
and $W=w_2w_4\ldots w_{2n}$. Then, the integrals can be written
\begin{equation}
\label{E:mkdvoddinte1}
\begin{split}
I^{\widetilde{\mk}}_r=\alpha_1\left( W^{-1}\left(\t^{1,2n}_{2r-1,0}+\t^{1,2n}_{2r-2,1}\right)+
W\left(\t^{1,2n}_{2r-1,1}+\t^{1,2n}_{2r-2,0}\right)\right)
\\
+ \alpha_2\left(\t^{1,2n}_{2r,1}+\t^{1,2n}_{2r-1,0}\right)+\alpha_3\left(\t^{1,2n}_{2r,0}+\t^{1,2n}_{2r-1,1}\right).
\end{split}
\end{equation}
where $\t=\t[f_i]$ with $f=K_p(w)$.

The map (\ref{E:mkdvoddmap}) has a symplectic structure $\Omega^{\widetilde{\mk}}_{2n}$, where
\begin{equation}
\label{E:mkdvoddsym}
\Omega^{\widetilde{\mk}}_{p}=\left(
\begin{matrix}
0& w_1w_2&0&0&\ldots   &0\\
-w_2w_1&0& w_2w_3&0&\ldots  &0\\
0&-w_3w_2&0& w_3w_4 &\ldots &0\\
\vdots&\vdots&\vdots&\vdots&\cdots&\vdots\\
0&0&0 &\ldots&0&w_{p-1}w_p\\
0&0&0& \ldots&-w_{p}w_{p-1}&0
 \end{matrix}
\right).
\end{equation}
This gives us a Poisson bracket $\{g,h\}_w =\nabla_w (g) \Omega_{2n}^{\widetilde{\mk}} \left(\nabla_w (h)\right)^{T}$.
Once again we can express the $w$-Poisson brackets between Theta multi-sums in terms of the
corresponding $f$-Poisson brackets.
\begin{lemma} \label{PBR3}
With $g,h$ differential functions on $\R^{p}$ we have
\[
\{ g \circ K_p , h\circ K_p\}_w = \{ g , h \}_{f=K_p(w)},
\]
i.e. $K_p$ is a Poisson map.
\end{lemma}
\begin{proof}
This follows from
\begin{equation}
\label{E:mkdvoddrela}
dK_p\Omega^{\widetilde{\mk}}_p dK_p^T = \Omega^{\sg}_{p} \mid_{f=K_p(w)}.
\end{equation}
\end{proof}
Because $F^{\pm 1} \circ K_{2n} = W^{\pm 1}$ this Lemma implies that $\{W,W^{-1}\}_w=0$,
\[
\{W^{\pm 1},\Theta^{1,2n}_{r,\epsilon}\}_w=\left\{
\begin{array}{ll}
0&\ \mbox {if} \ r \ \mbox{even},\\
\mp (-1)^{\epsilon} W^{\pm 1} \t^{1,2n}_{r,\epsilon}&\  \mbox{if}\ r\ \mbox{odd},
\end{array}
\right. 
\]
and we can also evaluate the brackets between $\t^{1,2n}_{r,\epsilon}$. Thus, the following
theorem can be proven by mechanical expansion and evaluation of the bracket.
\begin{theorem}
 Let $I^{\widetilde{\mk}}_r$ and $I^{\widetilde{\mk}}_s$ be given by the formula~(\ref{E:mkdvoddinte}). Then
 \[
 \{I^{\widetilde{\mk}}_r,I^{\widetilde{\mk}}_s\}_w=0.
 \]
\end{theorem}

\section{Involutivity of  pKdV integrals\label{S:pKinvo}}
In this section, we prove the involutivity of the integrals of order-reduced pKdV maps.
Similar to the sine-Gordon map, we consider two cases where the dimension $d$ of the map (\ref{E:pk}) is even or odd.
Here, in both cases there are not enough integrals for integrability, and therefore we perform reductions.
We present symplectic structures for the reduced maps in both cases and give a relationship between these symplectic structures.
For the case where $d$ is even, properties of multi-sums of products, $\Psi$, with respect to its symplectic structure
are proved in Appendix B. For the  case where $d$ is odd, the Poisson bracket between $\Psi$ multi-sums are derived from those
in the even case and the relationship between the two symplectic structures.
 \subsection{The case $d=2n+2$ \label{SS:pkodd}\label{SS:pkdveven}}
 We have  a $(2n+2)$-dimensional map~(\ref{E:map}). The integrals $I_r$ of this map
 are given by (\ref{E:pkinter}) with $0\leq r\leq n-1$ which are not enough integrals for integrability in
 the sense of Liouville-Arnold. Therefore, we use a reduction $c_i=v_{i}-v_{i+2}$ to reduce the dimension
 of the map by 2. From equation (\ref{E:pk}), we obtain the following map:
 \begin{equation}
 \label{E:pkoddmap}
 \pk: \ (c_1,c_2,\ldots,c_{2n})\mapsto (c_2,c_3,\ldots, c_{2n},\frac{\gamma}{c_2+c_4+\ldots + c_{2n}}-c_1-c_3-\ldots -c_{2n-1}).
 \end{equation}
 This map has exactly $n$ integrals given by
 \begin{align}
I^{\pk}_r=&\Psi^{1,2n-1}_{r-1}-(c_2+c_4+\ldots +c_{2n})\Psi^{1,2n-2}_{r-1}-(c_1+c_3+\ldots+c_{2n-1})\Psi^{2,2n-1}_{r-1}\notag\\
&+\Psi^{2,2n-2}_{r-2}+\left((c_1+c_3+\ldots +c_{2n-1})(c_2+c_4+\ldots +c_{2n})-\gamma \right)\Psi^{1,2n-1}_r, \label{E:pkoddinte}                                              \end{align}
with $r=0,1,\ldots,n-1$. The map is symplectic, we have $d\pk \cdot \Omega^{\pk}_{2n} \cdot d\pk^T = \Omega^{\pk}_{2n} \circ \pk$,
where
\begin{equation}
\label{E:pkoddsym}
 \Omega^{\pk}_{p}=
 \left(
 \begin{matrix}
 0&1&0&0\ldots& 0\\
 -1&0&1&0\ldots &0\\
 0&-1&0&1\ldots &0\\
 \vdots&\vdots&\vdots&\vdots&\vdots\\
 0&0&\ldots& 0&1\\
 0&0&\ldots &-1&0
 \end{matrix}
 \right) ,
 \end{equation}
which is given in \cite{Iatrou2003Higher-Dimension,RQs}. The corresponding Poisson bracket is denoted
$\{g,h\}_c =\nabla_c (g) \Omega_{2n}^{\pk} \left(\nabla_c (h)\right)^{T}$. We prove that the integrals of the map $\pk$ are in
involution with respect to this Poisson bracket. The proof is based on knowledge of the Poisson brackets between
two $\Psi$ multi-sums which is given as follows.

\begin{lemma}
\label{L:pkodd}
Let $p\geq 1$ and $0\leq r,s\leq \lfloor (p+1)/2 \rfloor$. Then  we have the  following identities
\begin{align}
\{\Psi^{1,p}_r,\Psi^{1,p}_s\}_c&=0,\label{E:PsiPoi1}\\
\{\Psi^{1,p}_r,\Psi^{1,p-1}_s\}_c+\{\Psi^{1,p-1}_r,\Psi^{1,p}_s\}_c&=0\label{E:PsiPoi2}.
\end{align}
\end{lemma}
  A corollary~\ref{C:pkodd} derives from this Lemma is given in Appendix B.

\begin{theorem}
\label{T:pkodd}
For all  $0\leq r,s \leq n-1$, we have $\{I_r,I_s\}_c=0$, where $I_r, I_s $ are given by~(\ref{E:pkoddinte}).
\end{theorem}
\begin{proof}
To prove this theorem we need the following formulas. Let $g(c_1,c_2,\ldots, c_{2n})$
 be a differentiable function on $\R^{2n}$.  Denote
 $$
 C_1=c_1+c_3+\cdots+c_{2n-1},\
 C_2=c_2+c_4+\cdots+c_{2n},
 $$
 we have
\begin{equation}
\label{E:pkpro}
\{g,C_1\}_c=-\frac{\partial g}{\partial c_{2n}},\
\{g,C_2\}_c=\frac{\partial g}{\partial c_1}.
\end{equation}
In addition, since we have
\[
\Psi^{a,b}_r=c_{b+1}\Psi^{a,b}_r+\Psi^{a,b-1}_{r-1}\ \mbox{and}\ \Psi^{a,b}_{r}=c_a\Psi^{a+1,b}_r+\Psi^{a+2,b}_{r-1},
\]
we obtain
\[
\displaystyle{\frac{\partial \Psi^{a,b+1}_r}{\partial c_{b+2}}=\Psi^{a,b}_r}\ \mbox{and}\
\frac{\partial \Psi^{a,b+1}_r}{\partial c_{a}}=\Psi^{a+1,b}_{r}.
\]
Now we write $\{I_r,I_s\}_c=A_1+A_2+A_3+A_4+A_5+A_6+A_7+A_8+A_9+A_{10}+A_{11}$, where
\begin{align*}
A_1:&=\{\Psi^{1,2n-1}_{r-1}-C_2\Psi^{1,2n-2}_{r-1},\Psi^{1,2n-1}_{s-1}-C_2\Psi^{1,2n-2}_{s-1}\}_c,\\
A_2:&=-\{\Psi^{1,2n-1}_{r-1},C_1\Psi^{2,2n-1}_{s-1}\}_c-\{C_1\Psi^{2,2n-1}_{r-1},\Psi^{1,2n-1}_{s-1}\}_c
+\{C_1\Psi^{2,2n-1}_{r-1},C_1\Psi^{2,2n-1}_{s-1}\}_c,\\
A_3:&=\{\Psi^{1,2n-1}_{r-1},\Psi^{2,2n-2}_{s-2}\}_c+\{\Psi^{2,2n-2}_{r-2},\Psi^{1,2n-1}_{s-1}\}_c+
\{\Psi^{2,2n-2}_{r-2},\Psi^{2,2n-2}_{s-2}\}_c,\\
A_4:&=\{C_2\Psi^{1,2n-2}_{r-1},C_1\Psi^{2,2n-1}_{s-1}\}_c + \{C_1\Psi^{2,2n-1}_{r-1},C_2\Psi^{1,2n-2}_{s-1}\}_c,\\
A_5:&=-\{C_2\Psi^{1,2n-2}_{r-1},\Psi^{2,2n-2}_{s-2}\}_c-\{\Psi^{2,2n-2}_{r-2},C_2\Psi^{1,2n-2}_{s-1}\}_c,\\
A_6:&=-\{C_1\Psi^{2,2n-1}_{r-1},\Psi^{2,2n-2}_{s-2}\}_c-\{\Psi^{2,2n-2}_{r-2},C_1\Psi^{2,2n-1}_{s-1}\}_c,\\
A_7:&=-\{\Psi^{1,2n-1}_{r-1},(C_1C_2-\gamma)\Psi^{1,2n-1}_s\}_c - \{(C_1C_2-\gamma)\Psi^{1,2n-1}_r,\Psi^{1,2n-1}_{s-1}\}_c,\\
A_8:&=-\{C_2\Psi^{1,2n-2}_{r-1},(C_1C_2-\gamma)\Psi^{1,2n-1}_s\}_c-\{(C_1C_2-\gamma)\Psi^{1,2n-1}_r,C_2\Psi^{1,2n-2}_{s-1}\}_c, \\
A_9:&=-\{C_1\Psi^{2,2n-1}_{r-1},(C_1C_2-\gamma)\Psi^{1,2n-1}_s\}_c- \{(C_1C_2-\gamma)\Psi^{1,2n-1}_r,C_1\Psi^{2,2n-1}_{s-1}\}_c,\\
A_{10}&:=\{\Psi^{2,2n-2}_{r-2},(C_1C_2-\gamma)\Psi^{1,2n-1}_s\}_c
 + \{(C_1C_2-\gamma)\Psi^{1,2n-1}_r,\Psi^{2,2n-2}_{s-2}\}_c,\\
A_{11}&:=\{(C_1C_2-\gamma)\Psi^{1,2n-1}_r,(C_1C_2-\gamma)\Psi^{1,2n-1}_s\}_c.
\end{align*}

Using Lemma \ref{L:pkodd}, Corollary \ref{C:pkodd} and formulas (\ref{E:pkpro}), we have
\begin{align*}
A_1&=\Psi^{1,2n-2}_{r-1}\Psi^{2,2n-1}_{s-1}-\Psi^{1,2n-2}_{s-1}\Psi^{2,2n-1}_{r-1}+
C_2\left(\Psi^{1,2n-2}_{s-1}\Psi^{2,2n-2}_{r-1}-\Psi^{1,2n-2}_{r-1}\Psi^{2,2n-2}_{s-1}\right),\\
A_2&=\Psi^{2,2n-1}_{s-1}\Psi^{1,2n-2}_{r-1}-\Psi^{2,2n-1}_{r-2}\Psi^{1,2n-2}_{s-1}+
C_1\left(\Psi^{2,2n-1}_{r-1}\Psi^{2,2n-2}_{s-1} - \Psi^{2,2n-1}_{s-1}\Psi^{2,2n-2}_{r-1}\right),\\
A_3&=\Psi^{2,2n-1}_{r-1}\Psi^{1,2n-2}_{s-1}-\Psi^{2,2n-1}_{s-1}\Psi^{1,2n-2}_{r-1},\\
A_4&=-\Psi^{1,2n-2}_{r-1}\Psi^{2,2n-1}_{s-1}+\Psi^{2,2n-1}_{r-1}\Psi^{1,2n-2}_{s-1},\\
A_5&=-C_2\left(\Psi^{1,2n-2}_{s-1}\Psi^{2,2n-2}_{r-1}-\Psi^{1,2n-2}_{r-1}\Psi^{2,2n-2}_{s-1}\right),\\
A_6&=-C_1\left(\Psi^{2,2n-1}_{r-1}\Psi^{2,2n-2}_{s-1}-\Psi^{2,2n-1}_{s-1}\Psi^{2,2n-2}_{r-1} \right).
\end{align*}
It follows that $A_1+A_2+A_3+A_4+A_5+A_6=0$. Now we show that $A_7+A_8+A_9+A_{10}+A_{11}=0$.
We also have
\begin{align*}
A_7&=C_1\left(\Psi^{1,2n-1}_s\Psi^{2,2n-1}_{r-1}-\Psi^{1,2n-1}_r\Psi^{2,2n-1}_{s-1}\right)
 +C_2\left(\Psi^{1,2n-1}_r\Psi^{1,2n-2}_{s-1}-\Psi^{1,2n-1}_s\Psi^{1,2n-2}_{r-1}\right),\\
A_8&=C_2(C_1C_2-\gamma)\left(\Psi^{1,2n-1}_s\Psi^{1,2n-2}_r-\Psi^{1,2n-1}_r\Psi^{1,2n-2}_s\right)+C_2
\left(\Psi^{1,2n-2}_{r-1}\Psi^{1,2n-1}_s-\Psi^{1,2n-2}_{s-1}\Psi^{1,2n-1}_r\right)\\
&\quad  +C_1C_2\left(\Psi^{1,2n-1}_r\Psi^{2,2n-2}_{s-1}-\Psi^{1,2n-1}_s\Psi^{2,2n-1}_{r-1}\right)+
 (C_1C_2-\gamma)\left(\Psi^{1,2n-2}_{r-1}\Psi^{2,2n-1}_s-\Psi^{1,2n-2}_{s-1}\Psi^{2,2n-1}_r\right),\\
A_9&=C_1(C_1C_2-\gamma)\left(\Psi^{1,2n-1}_r\Psi^{2,2n-1}_s-\Psi^{1,2n-1}_s\Psi^{2,2n-1}_r\right)+
C_1\left(\Psi^{2,2n-1}_{s-1}\Psi^{1,2n-1}_r-\Psi^{2,2n-1}_{r-1}\Psi^{1,2n-1}_s\right)\\
&\quad +C_1C_2\left(\Psi^{1,2n-1}_s\Psi^{2,2n-2}_{r-1}-\Psi^{1,2n-1}_r\Psi^{2,2n-2}_{s-1}\right)
 + (C_1C_2-\gamma)\left(\Psi^{2,2n-1}_{s-1}\Psi^{1,2n-2}_r-\Psi^{2,2n-1}_{r-1}\Psi^{1,2n-2}_s\right),\\
A_{10}&=(C_1C_2-\gamma)\left(\Psi^{2,2n-1}_r\Psi^{1,2n-2}_{s-1}-\Psi^{2,2n-1}_s\Psi^{1,2n-2}_{r-1} +\Psi^{2,2n-1}_{r-1}\Psi^{1,2n-2}_s-\Psi^{2,2n-1}_{s-1}\Psi^{1,2n-2}_{r}\right),\\
A_{11}&=(C_1C_2-\gamma)
C_1\left(\Psi^{1,2n-1}_s\Psi^{2,2n-1}_r-\Psi^{1,2n-1}_r\Psi^{2,2n-1}_s\right)+
\left(\Psi^{1,2n-2}_s\Psi^{1,2n-1}_r-\Psi^{1,2n-1}_s\Psi^{1,2n-2}_r \right)\\
&\quad (C_1c_2-\gamma)C_2.
\end{align*}
This implies $A_7+A_8+A_9+A_{10}+A_{11}=0$. Therefore, we have $\{I_r,I_s\}=0$.
\end{proof}
\subsection{The case  $d=2n+1$}
We introduce a reduction $u_i=v_i-v_{i+1}$. We obtain a $2n$-dimensional map
\begin{equation}
\label{E:pkevenmap}
\widetilde{\pk}: \ (u_1,u_1,\ldots, u_{2n})\mapsto (u_2,u_3,\ldots,u_{2n},\frac{\gamma}{u_2+u_3+\ldots+u_{2n}}-u_1-u_2-\ldots-u_{2n})
\end{equation}
with  $n$ integrals ($0\leq r\leq n-1$)
 \begin{align}
I^{\widetilde{\pk}}_r&=\Psi^{1,2n-2}_{r-1}-(u_2+u_3+\ldots +u_{2n})\Psi^{1,2n-3}_{r-1}-(u_1+u_2+\ldots+u_{2n-1})\Psi^{2,2n-2}_{r-1}\notag\\
&+\Psi^{2,2n-3}_{r-2}+\left((u_2+u_3+\ldots +u_{2n})(u_1+u_2+\ldots +u_{2n-1})-\gamma \right)\Psi^{1,2n-2}_r\label{E:pkeveninte},                                              \end{align}
where the argument of $\Psi$  is  $f_i=1/(c_ic_{i+1})$ with $c_i:=u_{i}+u_{i+1}$.
Based on the method given in \cite{RQs}, we obtain a symplectic structure $\Omega^{\widetilde{\pk}}_{2n}$
for the map~(\ref{E:pkevenmap}), where
 \begin{equation}
 \label{E:pkevensym}
 \Omega^{\widetilde{\pk}}_ p=
 \left(
 \begin{matrix}
 0&1&-1&1\ldots& (-1)^{p}\\
 -1&0&1&-1\ldots &-1^{p-1}\\
 1&-1&0&1\ldots &1\\
 \vdots&\vdots&\vdots&\vdots&\vdots\\
 (-1)^{p-1}&(-1)^{p-2}&\ldots&0&1\\
 (-1)^p&(-1)^{p-1}&\ldots &-1&0
 \end{matrix}
 \right).
 \end{equation}
The Poisson bracket is denoted $\{g,h\}_u =\nabla_u (g) \Omega_{2n}^{\widetilde{\pk}} \left(\nabla_u (h)\right)^{T}$.
Next we present a relationship between the two symplectic structures~(\ref{E:pkoddsym}) and~(\ref{E:pkevensym})
and the corresponding Poisson brackets. Consider the map
\[
Q_p: (u_1,u_2,\ldots, u_{p})\mapsto(u_1+u_2,u_2+u_3,\ldots,u_{p-1}+u_{p}).
\]
\begin{lemma}
The map $Q_p$ is a Poisson map, i.e.
\begin{equation}
\label{E:pkPsrela}
 \{f\circ Q_p,g\circ Q_p\}_u=\{f,g\}_{c=Q_p(u)},
\end{equation}
where $f(c)$ and $g(c)$ are differentiable functions.
\end{lemma}
\begin{proof}
By calculation we obtain
\begin{equation}
\label{E:pksymrela}
dQ_p\Omega^{\widetilde{\pk}}_{p}dQ_p^{T}=\Omega^{\pk}_{p-1}.
\end{equation}
\end{proof}

 \begin{theorem}
 \label{T:pkeven}
Let $I_r, I_s$ be given by~(\ref{E:pkeveninte}). Then,  for all $0\leq r,s\leq n-1$ we have
\[
\{I_r,I_s\}_u=0.
\]
 \end{theorem}
 \begin{proof}
 As the following formulas hold,
 \begin{align}
\{g,u_2+u_3+\ldots+u_{2n}\}_u&=\frac{\partial g}{\partial u_1},\label{E:pkpro4}\\
\{g,u_1+u_2+\ldots +u_{2n-1}\}_u&=-\frac{\partial g}{\partial u_{2n}},\label{E:pkpro5}
\end{align}
and the properties of Psi with respect to the bracket $\{,\}_u$ which are the same as those
with respect to the bracket $\{,\}_c$, one can prove the involutivity of the integrals~(\ref{E:pkeveninte})
similarly to what we did for the case $d=2n+2$.
 \end{proof}
 \section{Discussion\label{S:disc}}
 In this paper, we have proved the involutivity of  integrals of  sine-Gordon, pKdV and mKdV maps directly by
using induction and  using recently found symplectic structures of  these maps. In order to prove these maps are completely integrable in
the sense of Louville-Arnold \cite{Bruschi1991,Ves}, we also need to prove functional independence of their integrals which we hope to
publish  elsewhere \cite{Tranfun}.

It should be noted that the integrals of maps obtained as $(p,-1)$-reductions of the equations in the ABS list \cite{ABS}, with
the exception of $Q_4$, can be expressed in terms of multi-sums of products, $\Psi$ \cite{Tranclosedform}. Therefore, it would be
interesting to study their symplectic structures and furthermore their complete integrability.

\section*{Acknowledgment\label{S:Acknow}}
This research has been funded by the Australian Research Council through the Centre
of Excellence for Mathematics and Statistics of Complex Systems.
DTT acknowledges the support of two scholarships, one from La Trobe University
and the other from the Endeavour IPRS programme.
\appendix
 \section{Properties of $\Theta$ with respect to the Poisson brackets}
 In this Appendix, we prove Lemma~\ref{L:PoiThetaeven} and Lemma \ref{L:ThetaPoieven2}.
 First of all, the following lemma follows  from a property of the operator $E_f$ \eqref{E:Ef_operator}.
 \begin{lemma}
\label{L:PoiThetaeven1}
\begin{equation}
\label{E:PoiTheta1}
\{\t^{1,p}_{r,\epsilon},f_{p+1}^{(-1)^{\delta}}\}_f=
\left\{
\begin{array}{ll}
0&\ r \ \mbox{even}\\
(-1)^{\delta+\epsilon+1}f_{p+1}^{(-1)^{\delta}}\t^{1,p}_{r,\epsilon} & \ r \ \mbox
{odd}
\end{array}
\right.
\end{equation}
\end{lemma}
\begin{proof}
It is because
\[
\{\t^{1,p}_{r,\epsilon},f_{p+1}^{(-1)^{\delta}}\}_f=(-1)^{\delta}f_{p+1}^{(-1)^{\delta}}E_f\t^{1,p}_{r,\epsilon}.
\]
\end{proof}

\begin{remark}
\label{R:PoiTransform}
If we introduce $t_i=1/f_i$, then we get
\[
\t^{1,p}_{r,\epsilon}[f_i]=\t^{1,p}_{r,\epsilon+1}[t_i].
\]
We have
\begin{align*}
\{\t^{1,p}_{r,\epsilon},\t^{1,p}_{s,\delta}\}_f&=\sum_{i<j}
\left(\frac{\partial \t^{1,p}_{r,\epsilon}}{\partial f_i }\frac{\partial \t^{1,p}_{s,\delta}}{\partial f_j }
-\frac{\partial \t^{1,p}_{r,\epsilon}}{\partial f_j }\frac{\partial \t^{1,p}_{s,\delta}}{\partial f_i }\right)f_if_j\\
&=\sum_{i<j}\left(\frac{\partial \t^{1,p}_{r,\epsilon+1}}{\partial t_i }\frac{\partial \t^{1,p}_{s,\delta+1}}{\partial t_j }t_i^2t_j^2
-\frac{\partial \t^{0,p}_{r,\epsilon+1}}{\partial t_j }\frac{\partial \t^{1,p}_{s,\delta+1}}{\partial t_i }t_i^2t_j^2\right)\frac{1}{t_it_j}\mid_{t=T(f)}\\
&=\{\t^{1,p}_{r,\epsilon+1},\t^{1,p}_{s,\delta+1}\}_{f=T(f)},
\end{align*}
where $T$ is defined as follows
\[
T:\left(f_1,f_2,\ldots,f_p\right)\mapsto\left(\frac{1}{f_1},\frac{1}{f_2},\ldots, \frac{1}{f_p}\right).
\]
\end{remark}

\subsection{Proof of  Lemma~\ref{L:PoiThetaeven}}
\begin{proof}
We will prove this lemma by induction.
The following properties, given in \cite{KRQ, Tranclosedform} will be used in our proof:
\begin{align}
\t^{a,b}_{r,\epsilon}&=\t^{a,b-1}_{r,\epsilon}+f_b^{(-1)^{\epsilon +r}}\t^{a,b-1}_{r-1,\epsilon},\label{E:recur1}\\
\t^{a,b}_{r,\epsilon}&=\t^{a+1,b}_{r,\epsilon}+f_a^{(-1)^{\epsilon\pm 1}}\t^{a+1,b}_{r-1,\epsilon\pm 1}.\label{E:recur2}
\end{align}

Using Remark~\ref{R:PoiTransform}, it is sufficient to prove for the case  $\epsilon=1$. One verifies that~(\ref{E:PoiTheta}) holds for $p=1,2$ and for all $1\leq r,s\leq p$.
Suppose that ~(\ref{E:PoiTheta}) holds for $p-1$ and $p$ ($p\geq 2$). We will prove
that~(\ref{E:PoiTheta}) holds for $p+1$.

Using identity~(\ref{E:recur1}), we expand the left hand side of~(\ref{E:PoiTheta}) and we obtain
\begin{align}
\label{E:sgPoieven1}
&\{\t^{1,p+1}_{r,1},\t^{1,p+1}_{s,1}\}_f\notag \\
=&\{\t^{1,p}_{r,1},\t^{1,p}_{s,1}\}_f+f_{p+1}^{(-1)^{r+1}}\{\t^{1,p}_{r-1,1},\t^{1,p}_{s,1}\}_f+
\t^{1,p}_{r-1,1}\{f_{p+1}^{(-1)^{p+1}},\t^{1,p}_{s,1}\}_f
 +f_{p+1}^{(-1)^{s+1}}\{\t^{1,p}_{r,1},\t^{1,p}_{s-1,1}\}_f\notag\\
& +\t^{1,p}_{s-1,1}\{\t^{1,p}_{r,1},f_{p+1}^{(-1)^{s+1}}\}_f
+f_{p+1}^{(-1)^{r+1}+(-1)^{s+1}}\{\t^{1,p}_{r-1,1},\t^{1,p}_{s-1,1}\}_f+
\t^{1,p}_{r-1,1}\t^{1,p}_{s-1,1}\{f_{p+1}^{(-1)^{r+1}},f_{p+1}^{(-1)^{s+1}}\}_f\notag\\
& +f_{p+1}^{(-1)^{s+1}}\t^{1,p}_{r-1,1}\{f_{p+1}^{(-1)^{r+1}},\t^{1,p}_{s-1,1}\}_f
+f_{p+1}^{(-1)^{r+1}}\t^{1,p}_{s-1,1}\{\t^{1,p}_{r-1,1},f_{p+1}^{(-1)^{s+1}}\}_f.
\end{align}
The case $r=s$ is trivial. Now we distinguish 3 cases.
\begin{enumerate}
\item  $r$ and $s$ are both even or both odd.
Since $ \{\t^{1,p+1}_{r,1},\t^{1,p+1}_{s,1}\}_f=-\{\t^{1,p+1}_{s,1},\t^{1,p+1}_{r,1}\}_f$, without  loss
of generality we assume that $r>s$.\\
If both $r$ and $s$ are even, on the right hand side of~(\ref{E:sgPoieven1}) the first, third, fifth, sixth,
seventh terms vanish. Thus, we have
\begin{align*}
\{\t^{1,p+1}_{r,1},\t^{1,p+1}_{s,1}\}_f=&f_{p+1}^{-1}\left(\{\t^{1,p}_{r-1,1},\t^{1,p}_{s,1}\}_f+\{\t^{1,p}_{r,1},\t^{1,p}_{s-1,1}\}_f \right)
+ f_{p+1}^{-1}\t^{1,p}_{r-1,1}\{f_{p+1}^{-1},\t^{1,p}_{s-1,1}\}_f\\
&
+ f_{p+1}^{-1}\t^{1,p}_{s-1,1}\{\t^{1,p}_{r-1,1},f_{p+1}^{-1}\}_f\\
=&f_{p+1}^{-1}\left(\sum_{i\geq 1}(-1)^{i}\t^{1,p}_{s-i,1}\t^{1,p}_{r+i-1,1}
+\sum_{i\geq 0}(-1)^i\t^{1,p}_{r+i,1}\t^{1,p}_{s-1-i,1}\right)
\\
&
+f_{p+1}^{-2}\t^{1,p}_{r-1,1}\t^{1,p}_{s-1,1} -f_{p+1}^{-2}\t^{1,p}_{s-1,1}\t^{1,p}_{r-1,1}\\
=& f_{p+1}^{-1}\left(\sum_{j\geq 0}(-1)^{j+1}\t^{1,p}_{s-j-1,1}\t^{1,p}_{r+j,1}+
\sum_{i\geq 0}(-1)^i\t^{1,p}_{r+i,1}\t^{1,p}_{s-1-i,1}\right)\\
=& 0.
\end{align*}
If both  $r$ and  $s$ are odd and assuming $r>s$, on the right hand side of~(\ref{E:sgPoieven1}) the first,
sixth, seventh, eighth, and ninth terms vanish. Therefore, we have
\begin{align*}
&\{\t^{1,p+1}_{r,1},\t^{1,p+1}_{s,1}\}_f\\
=& f_{p+1}\left(\{\t^{1,p}_{r-1,1},\t^{1,p}_{s,1}\}_f+\{\t^{1,p}_{r,1},\t^{1,p}_{s-1,1}\}_f \right)+
\t^{1,p}_{r-1,1}\{f_{p+1},\t^{1,p}_{s,1}\}_f+\t^{1,p}_{s-1,1}\{\t^{1,p}_{r,1},f_{p+1}\}_f\\
=&
f_{p+1}\left(\sum_{i\geq 0}(-1)^i\t^{1,p}_{r-1+i,1}\t^{1,p}_{s-i,1}-\sum_{i\geq 1}(-1)^{i-1}\t^{1,p}_{s-1-i,1}\t^{1,p}_{r+i,1}\right) -f_{p+1}\t^{1,p}_{r-1,1}\t^{1,p}_{s,1}+f_{p+1}\t^{1,p}_{s-1,1}\t^{1,p}_{r,1}\\
=& f_{p+1}\left(\sum_{i\geq 0}(-1)^i\t^{1,p}_{r-1+i,1}\t^{1,p}_{s-i,1}-\sum_{j\geq 2}(-1)^j\t^{1,p}_{s-j,1}\t^{1,p}_{r+j-1,1}
-\t^{1,p}_{r-1,1}\t^{1,p}_{s,1} +\t^{1,p}_{s-1,1}\t^{1,p}_{r,1}\right)\\
=&f_{p+1}\left(\t^{1,p}_{r-1,1}\t^{1,p}_{s,1}-\t^{1,p}_{r,1}\t^{1,p}_{s-1,1}
-\t^{1,p}_{r-1,1}\t^{1,p}_{s,1} +\t^{1,p}_{s-1,1}f_{p+1}\t^{1,p}_{r,1}\right)\\
=&0.
\end{align*}
\item
$r$ is even, $s$ is odd and $r>s$. We have
\begin{align*}
&\{\t^{1,p+1}_{r,1},\t^{1,p+1}_{s,1}\}_f\\
=& \{\t^{1,p}_{r,1},\t^{1,p}_{s,1}\}_f+\{\t^{1,p}_{r-1,1},\t^{1,p}_{s-1,1}\}_f
+f_{p+1}^{-1}\t^{1,p}_{s-1,1}\{\t^{1,p}_{r-1,1},f_{p+1}\}_f+\t^{1,p}_{r-1,1}\{f_{p+1}^{-1},\t^{1,p}_{s,1}\}_f\\
=&\sum_{i\geq 0}(-1)^i\t^{1,p}_{r+i,1}\t^{1,p}_{s-i,1}-\left(\sum_{i\geq 1}(-1)^{i-1}\t^{1,p}_{s-i-1,1}\t^{1,p}_{r-1+i,1}\right)
+\t^{1,p}_{s-1,1}\t^{1,p}_{r-1,1}+f_{p+1}^{-1}\t^{1,p}_{r-1,1}\t^{1,p}_{s,1}\\
=&\sum_{i\geq 0}(-1)^i\t^{1,p}_{r+i,1}\t^{1,p}_{s-i,1}+\sum_{i\geq 0}(-1)^{i}\t^{1,p}_{s-i-1,1}\t^{1,p}_{r-1+i,1}
 +f_{p+1}^{-1}\t^{1,p}_{r-1,1}\t^{1,p}_{s,1}\\
 =&\sum_{i\geq 0}(-1)^i\Theta^{1,p+1}_{r+i,1}\Theta^{1,p+1}_{s-i,1},
\end{align*}
where in the last step we used (\ref{E:recur1}).
\item $r$ is even, $s$ is odd and $r<s$. We do similarly as in the previous case.
Therefore, with $\epsilon=1$ identity~(\ref{E:PoiTheta}) holds for $p+1$. Then, it holds for all $p\geq 0$.
\end{enumerate}
\end{proof}

\subsection{Proof of Lemma~\ref{L:ThetaPoieven2}}
\begin{proof}
The proof proceeds  by induction again.
It is easy to see that~(\ref{E:ThetaPoieveniden1})
can be  rewritten  as  follows
\begin{equation}
\label{EA:ThePoieven01}
\{\t^{1,p}_{r,0},\t^{1,p}_{s,1}\}_f=\left\{
\begin{array}{ll}
\sum_{i\geq0}\left(\Theta^{1,p}_{r-2i-1,0}\Theta^{1,p}_{s+2i+1,1}-\Theta^{1,p}_{r-2i-1,1}\Theta^{1,p}_{s+2i+1,0}\right)&
\ r\leq s,\\
\sum_{i\geq 0}\left(\Theta^{1,p}_{s-2i-1,0}\Theta^{1,p}_{r+2i+1,1}-\Theta^{1,p}_{s-2i-1,1}\Theta^{1,p}_{r+2i+1,0}\right)&
 \  r>s.
 \end{array}
 \right.
\end{equation}
 Identities~(\ref{E:ThetaPoieveniden1}) (or~(\ref{EA:ThePoieven01})) and~(\ref{E:ThetaPoieveniden2}) hold for $p=1,2$.
Suppose that they hold for  $p-1$ and $p$ ($p\geq 2$). We will prove that they hold for $p+1$.

Using  identity~(\ref{E:recur1}), expanding the left hand sides of~(\ref{E:ThetaPoieveniden1})
(or~(\ref{EA:ThePoieven01})) and~(\ref{E:ThetaPoieveniden2}),
we have
\begin{align}
\label{E:SgLinduc1}
& \{\Theta^{1,p+1}_{r,0},\Theta^{1,p+1}_{s,1}\}_f\notag\\
=& \{\Theta^{1,p}_{r,0}+f_{p+1}^{(-1)^r}\Theta^{1,p}_{r-1,0},\Theta^{1,p}_{s,1}+f_{p+1}^{(-1)^{s+1}}\Theta^{1,p}_{s-1,1}\}_f\notag\\
=& \{\Theta^{1,p}_{r,0},\Theta^{1,p}_{s,1}\}_f+f_{p+1}^{(-1)^r+(-1)^{s+1}}\{\Theta^{1,p}_{r-1,0},\Theta^{1,p}_{s-1,1}\}_f
+\t^{1,p}_{r,0}\t^{1,p}_{s,1}\{f_{p+1}^{(-1)^r},f_{p+1}^{(-1)^{s+1}}\}_f\notag\\
&+ f_{p+1}^{(-1)^{s+1}}\{\Theta^{1,p}_{r,0},\t^{1,p}_{s-1,1}\}_f+f_{p+1}^{(-1)^r}\{\t^{1,p}_{r-1,0},\t^{1,p}_{s,1}\}_f
+\t^{1,p}_{s-1,1}\{\t^{1,p}_{r,0},f_{p+1}^{(-1)^{s+1}}\}_f\notag\\
&+ \t^{1,p}_{r-1,0}\{f_{p+1}^{(-1)^r},\t^{1,p}_{s,1}\}_f+ f_{p+1}^{(-1)^r}\t^{1,p}_{s-1,1}\{\t^{1,p}_{r-1,0},f_{p+1}^{(-1)^{s+1}}\}_f
+f_{p+1}^{(-1)^{s+1}}\t^{1,p}_{r-1,0}\{(f_{p+1}^{(-1)^r},\t^{1,p}_{s-1,1}\}_f.
\end{align}
\begin{enumerate}
\item  $r\equiv s\pmod 2$.
We distinguish 2 cases.\\
Case 1:    $s-r=k\geq 0$,  we first prove the following
\begin{equation}
\label{E:sGsumprop1}
\sum_{i\geq 0}(-1)^{i+\epsilon}\t^{1,p}_{s-1-i,i+\epsilon}\t^{1,p}_{r+i,i+\epsilon+1}
=\sum_{i\geq 0}(-1)^{i+\epsilon}\t^{1,p}_{r-i-1,i+\epsilon}\t^{1,p}_{s+i,i+\epsilon+1}.
\end{equation}
The left hand side of this identity equals
\begin{align*}
&\sum_{i=k}^{s-1}(-1)^{i+\epsilon}\t^{1,p}_{s-1-i,i+\epsilon}\t^{1,p}_{r+i,i+\epsilon+1}
+\sum_{i=0}^{k-1}(-1)^{i+\epsilon}\t^{1,p}_{s-1-i,i+\epsilon}\t^{1,p}_{r+i,i+\epsilon+1}\\
=&\sum_{j=0}^{r-1}(-1)^{k+j+\epsilon}\t^{1,p}_{r-j-1,k+j+\epsilon}\t^{1,p}_{s+j,k+j+\epsilon+1}\\
&
+\sum_{i=0}^{\frac{k}{2}-1}\left((-1)^{i+\epsilon}\t^{1,p}_{s-1-i,i+\epsilon}\t^{1,p}_{r+i,i+\epsilon+1}
+(-1)^{k-1-i+\epsilon}\t^{1,p}_{s-1-(k-i-1),k-i-1+\epsilon}\t^{1,p}_{r+k-i-1,k-i+\epsilon}\right)\\
=&\sum_{i=0}^{r-1}(-1)^{j+\epsilon}\t^{1,p}_{r-j-1,j+\epsilon}\t^{1,p}_{s+j,j+\epsilon+1}
+\sum_{i=0}^{\frac{k}{2}-1}\left((-1)^{i+\epsilon}\t^{1,p}_{s-1-i,i+\epsilon}\t^{1,p}_{r+i,i+\epsilon+1}+
(-1)^{i+\epsilon+1}\t^{1,p}_{r+i,i+\epsilon+1}\t^{1,p}_{s-i-1,i+\epsilon}\right)\\
=&\sum_{i=0}^{r-1}(-1)^{i+\epsilon}\t^{1,p}_{r-i-1,i+\epsilon}\t^{1,p}_{s+i,i+\epsilon+1},
\end{align*}
which is the right hand side of~(\ref{E:sGsumprop1}).

Now using~(\ref{E:recur1}),  we expand the right hand side of the first identity of~(\ref{EA:ThePoieven01}).
We have
\begin{align}
\label{E:SgRinduc1}
&\sum_{i\geq 0}\left(\t^{1,p+1}_{r-2i-1,0}\t^{1,p+1}_{s+2i+1,1}-\t^{1,p+1}_{r-2i-1,1}\t^{1,p+1}_{s+2i+1,0}\right)
\notag\\
=&\sum_{i\geq 0}\left(\t^{1,p}_{r-2i-1,0}\t^{1,p}_{s+2i+1,1}-\t^{1,p}_{r-2i-1,1}\t^{1,p}_{s+2i+1,0}
+\t^{1,p}_{r-2-2i,0}\t^{1,p}_{s+2i,1}-\t^{1,p}_{r-2-2i,1}\t^{1,p}_{s+2i,0}\right)\notag\\
&+ f_{p+1}^{(-1)^{r-1}}\sum_{i\geq 0}
\left(\t^{1,p}_{r-2-2i,0}\t^{1,p}_{s+2i+1,1}-\t^{1,p}_{r-2i-1,1}\t^{1,p}_{s+2i,0}\right)\notag\\
& + f_{p+1}^{(-1)^s}\sum_{i\geq 0}\left(\t^{1,p}_{r-2i-1,0}\t^{1,p}_{s+2i,1}-\t^{1,p}_{r-2-2i,1}\t^{1,p}_{s+2i+1,0}\right)\notag \\
=&\sum_{i\geq 0}\left(\t^{1,p}_{r-2i-1,0}\t^{1,p}_{s+2i+1,1}-\t^{1,p}_{r-2i-1,1}\t^{1,p}_{s+2i+1,0}
+\t^{1,p}_{r-2-2i,0}\t^{1,p}_{s+2i,1}-\t^{1,p}_{r-2-2i,1}\t^{1,p}_{s+2i,0}\right)\notag \\
&+  f_{p+1}^{(-1)^{s+1}}\sum_{i\geq 0}(-1)^{i-1}\t^{1,p}_{r-1-i,i+1}\t^{1,p}_{s+i,i}
+ f_{p+1}^{(-1)^s}\sum_{i\geq 0}(-1)^i\t^{1,p}_{r-1-i,i}\t^{1,p}_{s+i,i+1}.
\end{align}
If $r$ and $s$ are both even. Using~(\ref{E:SgLinduc1}) and the induction assumption, we have
\begin{align*}
& \{\Theta^{1,p+1}_{r,0},\Theta^{1,p+1}_{s,1}\}_f\\
= &\sum_{i\geq 0}\left(\t^{1,p}_{r-2i-1,0}\t^{1,p}_{s+2i+1,1}-\t^{1,p}_{r-2i-1,1}\t^{1,p}_{s+2i+1,0}\right)
+\sum_{i\geq 0}\left(\t^{1,p}_{r-2-2i,0}\t^{1,p}_{s+2i,1}-\t^{1,p}_{r-2-2i,1}\t^{1,p}_{s+2i,0}\right)\\
&
+f_{p+1}^{(-1)^{r}}\sum_{i\geq 0}(-1)^i\t^{1,p}_{s+i,i+1}\t^{1,p}_{r-1-i,i}+
\t^{1,p}_{s-1,1}\t^{1,p}_{r-1,0}-\t^{1,p}_{r-1,0}\t^{1,p}_{s-1,1}\\
&+f_{p+1}^{(-1)^{s+1}}\sum_{i\geq 0}(-1)^{i-1}\t^{1,p}_{s-1-i,i+1}\t^{1,p}_{r+i,i},
\end{align*}
which equals~(\ref{E:SgRinduc1}) by using~(\ref{E:sGsumprop1}) with $\epsilon=1$.

If $r$ and $s$ are both odd, we have
\begin{align*}
& \{\Theta^{1,p+1}_{r,0},\Theta^{1,p+1}_{s,1}\}_f\\
= &\sum_{i\geq 0}\left(\t^{1,p}_{r-2i-1,0}\t^{1,p}_{s+2i+1,1}-\t^{1,p}_{r-2i-1,1}\t^{1,p}_{s+2i+1,0}\right)
+\sum_{i\geq 0}\left(\t^{1,p}_{r-2-2i,0}\t^{1,p}_{s+2i,1}-\t^{1,p}_{r-2-2i,1}\t^{1,p}_{s+2i,0}\right)\\
&+ f_{p+1}^{(-1)^{s+1}}\sum_{i\geq 0}(-1)^i\t^{1,p}_{s-1+i,i+1}\t^{1,p}_{r-i,i}+
f_{p+1}^{(-1)^{r}}\sum_{i\geq 0}(-1)^{i-1}\t^{1,p}_{s-i,i+1}\t^{1,p}_{r-1+i,i}
-f_{p+1}^{(-1)^{s+1}}\t^{1,p}_{s-1,1}\t^{1,p}_{r,0}
\\
&+f_{p+1}^{(-1)^{r}}\t^{1,p}_{r-1,0}\t^{1,p}_{s,1}\\
=&\sum_{i\geq 0}\left(\t^{1,p}_{r-2i-1,0}\t^{1,p}_{s+2i+1,1}-\t^{1,p}_{r-2i-1,1}\t^{1,p}_{s+2i+1,0}\right)
+\sum_{i\geq 0}\left(\t^{1,p}_{r-2-2i,0}\t^{1,p}_{s+2i,1}-\t^{1,p}_{r-2-2i,1}\t^{1,p}_{s+2i,0}\right)\\
&+f_{p+1}^{(-1)^{s+1}}\sum_{i\geq 0}(-1)^{i-1}\t^{1,p}_{s+i,i}\t^{1,p}_{r-1-i,i+1}
+f_{p+1}^{(-1)^{r}}\sum_{i\geq 0}(-1)^i\t^{1,p}_{s-1-i,i}\t^{1,p}_{r+i,i+1}.
\end{align*}
which equals~(\ref{E:SgRinduc1}) by using~(\ref{E:sGsumprop1}) with $\epsilon=0$.
Thus, the first identity of (\ref{E:ThetaPoieveniden1}) (or~(\ref{EA:ThePoieven01})) holds for $p+1.$

 Case 2: $s-r=-k<0$,  identity (\ref{E:ThetaPoieveniden1}) (or~(\ref{EA:ThePoieven01})) also holds  by using Remark~\ref{R:PoiTransform}.
\item $r\not\equiv s\pmod 2$.
Case 1: $s-r=k>0$
 With $r$ even, $s$ odd and $r<s$, we now expand the right hand side of the first identity of~(\ref{E:ThetaPoieveniden2})
with $p+1$.
Similar as~(\ref{E:sGsumprop1}), we have the following identities
\begin{align}
 \sum_{i\geq 0}(-1)^{i-1}\t^{1,p}_{s-i,i+1}\t^{1,p}_{r+i,i}&=\sum_{i\geq 1}(-1)^i\t^{1,p}_{s-1-i,i+1}\t^{1,p}_{r+i-1,i}
 \label{E:Sgsumprop2}\\
 \sum_{i\geq 0}(-1)^{i-1}\Theta^{1,p}_{s-i-1,i+1}\t^{1,p}_{r+i,i}+
\sum_{i\geq 0}(-1)^{i-1}\t^{1,p}_{s-i,i+1}\t^{1,p}_{r+i-1,i}
&=
\sum_{i\geq 1}(-1)^i\t^{1,p}_{r-1-i,i}\t^{1,p}_{s+i,i+1}\notag\\
&+
\sum_{i\geq 0}(-1)^{i-1}\t^{1,p}_{r-1-i,i+1}\t^{1,p}_{s+i,i}.
\label{E:Sgsumprop3}
\end{align}
We now expand the right hand side of~(\ref{E:ThetaPoieveniden2}) by using~(\ref{E:recur1}, we obtain
 \begin{align}
 \label{E:sGRinduc2}
 \sum_{i\geq 0}(-1)^{i-1}\t^{1,p+1}_{s-i,i+1}\t^{1,p+1}_{r+i,i}
 =& \sum_{i\geq 0}(-1)^{i-1}\t^{1,p}_{s-i,i+1}\t^{1,p}_{r+i,i} +
 f_{p+1}^2\sum_{i\geq 0}(-1)^{i-1}\t^{1,p}_{s-i-1,i+1}\t^{1,p}_{r+i-1,i}
 \notag\\
 & +
 f_{p+1}\sum_{i\geq 0}(-1)^{i-1}\Theta^{1,p}_{s-i-1,i+1}\t^{1,p}_{r+i,i}
 +f_{p+1}\sum_{i\geq 0}(-1)^{i-1}\t^{1,p}_{s-i,i+1}\t^{1,p}_{r+i-1,i}.
 \end{align}
For the left hand side of~(\ref{E:ThetaPoieveniden2}), using~(\ref{E:SgLinduc1}) and the induction assumption,
we have
 \begin{align*}
 &\{\t^{1,p+1}_{r,0},\t^{1,p+1}_{s,1}\}_f\\
 =& \sum_{i\geq 0}(-1)^{i-1}\t^{1,p}_{s-i,i+1}\t^{1,p}_{r+i,i}+
 f_{p+1}^2\sum_{i\geq 0}(-1)^i\t^{1,p}_{s-1+i,i+1}\t^{1,p}_{r-1-i,i}
 \\
 &+ f_{p+1}\sum_{i\geq 0}\left(\t^{1,p}_{r-2i-1,0}\t^{1,p}_{s-1+2i+1,1}-\t^{1,p}_{r-2i-1,1}\t^{1,p}_{s-1+2i+1,0}\right)
 \\
 &
 +f_{p+1}\sum_{i\geq 0}\left(\t^{1,p}_{r-1-2i-1,0}\t^{1,p}_{s+2i+1,1}-\t^{1,p}_{r-1-2i-1,1}\t^{1,p}_{s+2i+1,0}\right)
 -f_{p+1}\t^{1,p}_{r-1,0}\t^{1,p}_{s,1}-f_{p+1}^2\t^{1,p}_{s-1,1}\t^{1,p}_{r-1,0}\\
 =& \sum_{i\geq 0}(-1)^{i-1}\t^{1,p}_{s-i,i+1}\t^{1,p}_{r+i,i}
 +f_{p+1}^2\sum_{i\geq 1}(-1)^i\t^{1,p}_{s-1+i,i+1}\t^{1,p}_{r-1-i,i}
 + f_{p+1}\sum_{i\geq 1}(-1)^i\t^{1,p}_{r-1-i,i}\t^{1,p}_{s+i,i+1}
 \\
 &+
 f_{p+1}\sum_{i\geq 0}(-1)^{i-1}\t^{1,p}_{r-1-i,i+1}\t^{1,p}_{s+i,i}\\
 =& \sum_{i\geq 0}(-1)^{i-1}\t^{1,p+1}_{s-i,i+1}\t^{1,p+1}_{r+i,i}
 \end{align*}
 where in the final equality we used~(\ref{E:Sgsumprop2}) and~(\ref{E:Sgsumprop3}).
That implies that the second identity of (\ref{E:ThetaPoieveniden2}) holds for $p+1.$

With $r$ odd, $s$ even and $r<s$, we do similarly to what we did in  the previous case.

Case 2: $s=r=-k<0$, identity~(\ref{E:ThetaPoieveniden2}) still holds by using Remark~\ref{R:PoiTransform}.
\end{enumerate}
\end{proof}
\section{Properties of Psi  with respect to the Poisson brackets\label{AS:Psiprop}}
\subsection{Proof of Lemma~\ref{L:pkodd}}
\begin{proof}
We prove~(\ref{E:PsiPoi1}) and~(\ref{E:PsiPoi2}) simultaneously by induction.
We use the following property
\begin{equation}
\label{E:Psirecur}
\Psi^{a,b+1}_r=c_{b+2}\Psi^{a,b}_r+\Psi^{a,b-1}_{r-1}
\end{equation}
and therefore we get
\begin{equation}
\label{E:Psidrv}
\displaystyle{\frac{\partial \Psi^{a,b+1}_r}{\partial c_{b+2}}=\Psi^{a,b}_r}.
\end{equation}
We see that (1) and (2) hold for $p=1, 2, 3$.
Suppose (1) and (2) hold for $ p-2, p-1$ and $p$. We need to prove that (1) and (2) hold for $p+1$.
Expanding the right hand side of the first identity, we have
\begin{align*}
\{\Psi^{1,p+1}_r,\Psi^{1,p+1}_s\}_c&=\{c_{p+2}\Psi^{1,p}_r+\Psi^{1,p-1}_{r-1},c_{p+2}\Psi^{1,p}_s+\Psi^{1,p-1}_{s-1}\}_c\\
&=\{c_{p+2}\Psi^{1,p}_r,c_{p+2}\Psi^{1,p}_s\}_c+\{c_{p+2}\Psi^{1,p}_r,\Psi^{1,p-1}_{s-1}\}_c+
\{\Psi^{1,p-1}_{r-1},c_{p+2}\Psi^{1,p}_s\}_c+\{\Psi^{1,p-1}_{r-1},\Psi^{1,p-1}_{s-1}\}_c\\
&=c_{p+2}\Psi^{1,p}_s\{\Psi^{1,p}_r,c_{p+2}\}_c+c_{p+2}\Psi^{1,p}_r\{c_{p+2},\Psi^{1,p}_s\}_c
+\Psi^{1,p}_r\{c_{p+2},\Psi^{1,p-1}_s\}_c\\
&\quad +\Psi^{1,p}_s\{\Psi^{1,p-1}_r,c_{p+2}\}_c+ c_{p+2}\left(\{\Psi^{1,p}_r,\Psi^{1,p-1}_{s-1}\}_c
+\{\Psi^{1,p-1}_{r-1},\Psi^{1,p}_s\}_c\right)\\
&=c_{p+2}\left(\Psi^{1,p}_s\Psi^{1,p-1}_r-\Psi^{1,p}_r\Psi^{1,p-1}_s + \{\Psi^{1,p}_r,\Psi^{1,p-1}_{s-1}\}_c
+\{\Psi^{1,p-1}_{r-1},\Psi^{1,p}_s\}_c\right).
\end{align*}
 For the second identity, we also have
\begin{align*}
\{\Psi^{1,p+1}_r,\Psi^{1,p}_s\}_c+\{\Psi^{1,p}_r,\Psi^{1,p+1}_s\}_c&=\{c_{p+2}\Psi^{1,p}_r+\Psi^{1,p-1}_{r-1},\Psi^{1,p}_s\}_c+
\{\Psi^{1,p}_r,c_{p+2}\Psi^{1,p}_s+\Psi^{1,p-1}_{s-1}\}_c\\
&=\{\Psi^{1,p-1}_{r-1},\Psi^{1,p}_s\}_c+\{\Psi^{1,p}_r,\Psi^{1,p-1}_{s-1}\}_c
-\Psi^{1,p}_r\Psi^{1,p-1}_{s-1}+\Psi^{1,p}_s\Psi^{1,p-1}_{r-1}
\end{align*}
Now to prove (1) and (2) hold for $p+1$, we only need to prove that
 $$T:=\{\Psi^{1,p-1}_{r-1},\Psi^{1,p}_s\}_c+\{\Psi^{1,p}_r,\Psi^{1,p-1}_{s-1}\}_c
-\Psi^{1,p}_r\Psi^{1,p-1}_{s-1}+\Psi^{1,p}_s\Psi^{1,p-1}_{r-1}=0.$$
Using~(\ref{E:Psirecur}) and the induction assumption to expand $T$, we obtain
\begin{align*}
T&=(c_{p+1}\Psi^{1,p-1}_s+\Psi^{1,p-2}_{s-1})\Psi^{1,p-1}_r
-c_{p+2}(c_{p+1}\Psi^{1,p-1}_r+\Psi^{1,p-2}_{r-1})\Psi^{1,p-1}_s\\
& \quad + \{c_{p+1}\Psi^{1,p-1}_r+\Psi^{1,p-2}_{r-1},\Psi^{1,p-1}_{s-1}\}_c+
\{\Psi^{1,p-1}_{r-1},c_{p+1}\Psi^{1,p-1}_s+\Psi^{1,p-2}_{s-1}\}_c\\
&=\Psi^{1,p-2}_{s-1}\Psi^{1,p-1}_r-\Psi^{1,p-2}_{r-1}\Psi^{1,p-1}_{s}+
\left(\{\Psi^{1,p-2}_{r-1},\Psi^{1,p-1}_{s-1}\}_c+\{\Psi^{1,p-1}_{r-1},\Psi^{1,p-2}_{s-1}\}_c\right)\\
&\quad +\Psi^{1,p-1}_r\{c_{p+1},\Psi^{1,p-1}_{s-1}\}+\Psi^{1,p-1}_s\{\Psi^{1,p-1}_{r-1},c_{p+1}\}_c\\
&=\Psi^{1,p-2}_{s-1}\Psi^{1,p-1}_r-c_{p+2}\Psi^{1,p-2}_{r-1}\Psi^{1,p-1}_{s}+
\left(\{\Psi^{1,p-2}_{r-1},\Psi^{1,p-1}_{s-1}\}_c+\{\Psi^{1,p-1}_{r-1},\Psi^{1,p-2}_{s-1}\}_c\right)\\
&\quad -\Psi^{1,p-1}_r\Psi^{1,p-2}_{s-1}+\Psi^{1,p-1}_s\Psi^{1,p-2}_{r-1}\\
&=\left(\{\Psi^{1,p-2}_{r-1},\Psi^{1,p-1}_{s-1}\}+\{\Psi^{1,p-1}_{r-1},\Psi^{1,p-2}_{s-1}\}\right)\\
&=0.
\end{align*}
That means (1) and (2) hold for $p+1$.
\end{proof}
\subsection{Corollary of lemma~\ref{L:pkodd}}
We obtain the following formulas for (sums of) Poisson brackets of $\Psi$s.
\begin{corollary}
\label{C:pkodd}
\begin{enumerate}
Let $p\geq 1$ and let $r,s\in\Z$. Then,
\item $\{\Psi^{1,p}_r,\Psi^{1,p-1}_{s-1}\}_c+\{\Psi^{1,p-1}_{r-1},\Psi^{1,p}_s\}_c
=\Psi^{1,p}_r\Psi^{1,p-1}_s-\Psi^{1,p}_s\Psi^{1,p-1}_r,$
\item $\{\Psi^{a,b}_r,\Psi^{a,b}_s\}_c=0$ with $0\leq r,s\leq \lfloor(b-a)/2\rfloor+1$,
\item $\{\Psi^{1,p}_r,\Psi^{2,p}_{s-1}\}_c+\{\Psi^{2,p}_{r-1},\Psi^{1,p}_s\}_c
=\Psi^{1,n}_s\Psi^{2,p}_r-\Psi^{1,p}_r\Psi^{2,p}_s$,
\item $\{\Psi^{1,p}_r,\Psi^{2,p}_s\}_c+\{\Psi^{2,p}_r,\Psi^{1,p}_s\}_c=0$,
\item $\{\Psi^{1,p}_r,\Psi^{2,p+1}_s\}_c+\{\Psi^{2,p+1}_r,\Psi^{1,p}_s\}_c=0$,
\item $\{\Psi^{1,p}_r,\Psi^{2,p-1}_{s-1}\}_c+\{\Psi^{2,p-1}_{r-1},\Psi^{1,p}_s\}_c
=\Psi^{2,p}_r\Psi^{1,p-1}_s-\Psi^{2,p}_s\Psi^{1,p-1}_r$,
\item $\{\Psi^{1,p}_r,\Psi^{1,p-1}_{s-1}\}_c+\{\Psi^{1,p-1}_{r-1},\Psi^{1,p}_s\}_c=\Psi^{1,p}_r\Psi^{1,p-1}_s-\Psi^{1,p}_s\Psi^{1,p-1}_r$,
\item $\{\Psi^{1,p+1}_r,\Psi^{2,p}_s\}_c+\{\Psi^{2,p}_r,\Psi^{1,p+1}_s\}_c=0$,
\item $\{\Psi^{1,p}_r,\Psi^{2,p-1}_{s-2}\}_c+\{\Psi^{2,p-1}_{r-2},\Psi^{1,p}_s\}_c=
\Psi^{2,p}_r\Psi^{1,p-1}_{s-1}-\Psi^{2,p}_s\Psi^{1,p-1}_{r-1}+\Psi^{2,p}_{r-1}\Psi^{1,p-1}_s-\Psi^{2,p}_{s-1}\Psi^{1,p-1}_r$.
\end{enumerate}
\end{corollary}

\begin{proof}
\begin{enumerate}
\item It follows from the proof of Lemma~\ref{L:pkodd}.
\item It  follows from Lemma~\ref{L:pkodd}.
\item Using (2) we have $\{\Psi^{0,p}_r,\Psi^{0,p}_s\}_c=0$.
On the other hand, we have
\begin{align*}
\{\Psi^{0,p}_r,\Psi^{0,p}_s\}_c&=\{c_0\Psi^{1,p}_r+\Psi^{2,p}_{r-1},c_0\Psi^{1,p}_s+\Psi^{2,p}_{s-1}\}_c\\
&=\{c_0\Psi^{1,p}_r,c_0\Psi^{1,p}_s\}_c+\{c_0\Psi^{1,p}_s,\Psi^{2,p}_{s-1}\}_c+\{\Psi^{2,p}_{r-1},c_0\Psi^{1,p}_s\}_c
+\{\Psi^{2,p}_{r-1},\Psi^{2,p}_{s-1}\}_c\\
&=c_0\Psi^{1,p}_s\{\Psi^{1,p}_r,c_0\}+c_0\Psi^{1,p}_r\{c_0,\Psi^{1,p}_s\}_c+
c_0\left(\{\Psi^{1,p}_r,\Psi^{2,p}_{s-1}\}_c+\{\Psi^{2,p}_{r-1},\Psi^{1,p}_s\}_c\right)\\
&\quad + \Psi^{1,p}_s\{\Psi^{2,p}_{r-1},c_0\}_c+\Psi^{1,p}_r\{c_0,\Psi^{2,p}_{s-2}\}_c\\
&=-c_0\Psi^{1,p}_s\Psi^{2,p}_r+c_0\Psi^{1,p}_r\Psi^{2,p}_s+
c_0\left(\{\Psi^{1,p}_r,\Psi^{2,p}_{s-1}\}+\{\Psi^{2,p}_{r-1},\Psi^{1,p}_s\}\right)
\end{align*}
Therefore, we get
$\{\Psi^{1,p}_r,\Psi^{2,p}_{s-1}\}_c+\{\Psi^{2,p}_{r-1},\Psi^{1,p}_s\}_c
=\Psi^{1,p}_s\Psi^{2,p}_r-\Psi^{1,p}_r\Psi^{2,p}_s$.
\item For this identity, we expand the left hand side. We have
\begin{align*}
\rm{LHS}&=\{\Psi^{1,p}_r,\Psi^{0,p}_{s+1}-c_0\Psi^{1,p}_{s+1}\}_c
+\{\Psi^{0,p}_{r+1}-c_0\Psi^{1,p}_{r+1},\Psi^{1,p}_s\}_c\\
&=\{\Psi^{1,p}_r,\Psi^{0,p}_{s+1}\}_c+\{\Psi^{0,p}_{r+1},\Psi^{1,p}_s\}_c-
\Psi^{1,p}_{s+1}\{\Psi^{1,p}_r,c_0\}_c-\Psi^{1,p}_{r+1}\{c_0,\Psi^{1,p}_s\}_c\\
&=\{\Psi^{1,p}_r,\Psi^{0,p}_{s+1}\}_c+\{\Psi^{0,p}_{r+1},\Psi^{1,p}_s\}_c
+\Psi^{1,p}_{s+1}\Psi^{2,p}_r-\Psi^{1,p}_{r+1}\Psi^{2,p}_s.
\end{align*}
By property (3) we have
\begin{align*}
\{\Psi^{1,p}_r,\Psi^{0,p}_{s+1}\}_c+\{\Psi^{0,p}_{r+1},\Psi^{1,p}_s\}_c&=
\Psi^{0,p}_{s+1}\Psi^{1,p}_{r+1}-\Psi^{0,p}_{r+1}\Psi^{1,p}_{s+1}\\
&=(c_0\Psi^{1,p}_{r+1}+\Psi^{2,p}_r)\Psi^{1,p}_{r+1}-(c_0\Psi^{1,p}_{r+1}+\Psi^{2,n}_r)\Psi^{1,p}_{s+1}\\
&=\Psi^{2,p}_s\Psi^{1,p}_{r+1}-\Psi^{2,p}_r\Psi^{1,p}_{s+1}.
\end{align*}
It means that $\{\Psi^{1,p}_r,\Psi^{2,p}_s\}_c+\{\Psi^{2,p}_r,\Psi^{1,p}_s\}_c=0$.
\item  One can verify that this identity holds for $p=1,2,3$. Expanding  the left hand side using~(\ref{E:Psirecur}), we have
 \begin{align*}
 LHS
 &=
 \{\Psi^{1,p}_r,c_{p+2}\Psi^{2,p}_s+\Psi^{2,p-1}_{s-1}\}_c+\{c_{p+2}\Psi^{2,p}_r+\Psi^{2,p-1}_{r-1},\Psi^{1,p}_s\}_c\\
 &=c_{p+2}(\{\Psi^{1,p}_r,\Psi^{2,p}_s\}_c+\{\Psi^{2,p}_r,\Psi^{1,p}_s\}_c)+
 \Psi^{2,p}_s\{\Psi^{1,p}_r,c_{p+2}\}_c+\Psi^{2,p}_r\{c_{p+2},\Psi^{1,p}_s\}_c\\
&\quad + \{\Psi^{1,p}_r,\Psi^{2,p-1}_{s-1}\}_c+\{\Psi^{2,p-1}_{r-1},\Psi^{1,p}_s\}_c\\
&=\Psi^{2,p}_s\Psi^{1,p-1}_r-\Psi^{2,p}_r\Psi^{1,p-1}_s+
\{c_{p+1}\Psi^{1,p-1}_r+\Psi^{1,p-2}_{r-1},\Psi^{2,p-1}_{s-1}\}_c
\\
&\quad
+\{\Psi^{2,p-1}_{r-1},c_{p+1}\Psi^{1,p-1}_s+\Psi^{1,p-2}_{s-1}\}_c\\
&=\Psi^{2,p}_s\Psi^{1,p-1}_r-\Psi^{2,p}_r\Psi^{1,p-1}_s+
c_{p+1}(\{\Psi^{1,p-1}_r,\Psi^{2,p-1}_{s-1}\}_c+\{\Psi^{2,p-1}_{r-1},\Psi^{1,p-1}_s\}_c)\\
&\quad + \Psi^{1,p-1}_r\{c_{p+1},\Psi^{2,p-1}_{s-1}\}_c+\Psi^{1,p-1}_s\{\Psi^{2,p-1}_{r-1},c_{p+1}\}_c
+\{\Psi^{1,p-2}_{r-1},\Psi^{2,p-1}_{s-1}\}_c+\{\Psi^{2,p-1}_{r-1},\Psi^{1,p-2}_{s-1}\}_c\\
&=\Psi^{2,p}_s\Psi^{1,p-1}_r-\Psi^{2,p}_r\Psi^{1,p-1}_s+c_{p+1}(\Psi^{1,p-1}_s\psi^{2,p-1}_r-\Psi^{1,p-1}_r\Psi^{2,p-1}_s)\\
&\quad-\Psi^{1,p-1}_r\Psi^{2,p-2}_{s-1}+\Psi^{1,p-1}_s\Psi^{2,p-2}_{r-1}
+\{\Psi^{1,p-2}_{r-1},\Psi^{2,p-1}_{s-1}\}_c+\{\Psi^{2,p-1}_{r-1},\Psi^{1,p-2}_{s-1}\}_c\\
&=\{\Psi^{1,p-2}_{r-1},\Psi^{2,p-1}_{s-1}\}_c+\{\Psi^{2,p-1}_{r-1},\Psi^{1,p-2}_{s-1}\}_c.
 \end{align*}
 Therefore, using induction we prove our statement.
 \item This identity follows from the proof of the identity in (5).
 \item We have
 \begin{align*}
\{\Psi^{1,p}_r,\Psi^{1,p-1}_{s-1}\}_c+\{\Psi^{1,p-1}_{r-1},\Psi^{1,p}_s\}_c&=
\{\Psi^{1,p}_r,\Psi^{1,p+1}_{s}-c_{p+2}\Psi^{1,p}_s\}_c
\\&
\quad +\{\Psi^{1,p+1}_r-c_{p+2}\Psi^{1,p}_r,\Psi^{1,p}_s\}_c\\
&=-\Psi^{1,p}_s\{\Psi^{1,p}_r,c_{p+2}\}_c-\Psi^{1,p}_r\{c_{p+2},\Psi^{1,p}_s\}_c\\
&=-\Psi^{1,p}_s\Psi^{1,p-1}_r+\Psi^{1,p}_r\Psi^{1,p-1}_s
 \end{align*}
 \item We prove this by induction. This identity holds for $p=1,2$. Suppose that this identity holds for
  $p-2$, and $p-1$ we need to prove that it holds for $p$.
 We have
 \begin{align*}
  &\quad \{\Psi^{1,p+1}_r,\Psi^{2,p}_s\}_c+\{\Psi^{2,p}_r,\Psi^{1,p+1}_s\}_c\\
  &=\{c_{p+2}\Psi^{1,p}_r+\Psi^{1,p-1}_{r-1},\Psi^{2,p}_s\}_c+\{\Psi^{2,p}_r,c_{p+2}\Psi^{1,p}_s+\Psi^{1,p-1}_{s-1}\}_c\\
  &=\{\Psi^{1,p-1}_{r-1},\Psi^{2,p}_s\}_c+\{\Psi^{2,p}_r,\Psi^{1,p-1}_{s-1}\}_c-\Psi^{1,p}_r\Psi^{2,p-1}_s+\Psi^{1,p}_s\Psi^{2,p-1}_r\\
  &=\{\Psi^{1,p-1}_{r-1},c_{p+1}\Psi^{2,p-1}_s+\Psi^{2,p-2}_{s-1}\}_c+\{c_{p+1}\Psi^{2,p-1}_r+\Psi^{2,p-2}_{r-1},\Psi^{1,p-1}_{s-1}\}_c
  \\
  &\quad
  -\Psi^{1,p}_r\Psi^{2,p-1}_s+\Psi^{1,p}_s\Psi^{2,p-1}_r\\
  &=\{\Psi^{1,p-1}_{r-1},\Psi^{2,p-2}_{s-2}\}_c+\{\Psi^{2,p-2}_{r-2},\Psi^{1,p-1}_{s-1}\}_c
  +c_{p+1}(\{\Psi^{1,p-1}_{r-1},\Psi^{2,p-1}_s\}_c+\{\Psi^{2,p-1}_r,\Psi^{1,p-1}_{s-1}\}_c)\\
  &\quad -\Psi^{2,p-1}_s\Psi^{1,p-2}_{r-1}+\Psi^{2,p-1}_r\Psi^{1,p-2}_{s-1}-\Psi^{1,p}_r\Psi^{2,p-1}_s+\Psi^{1,p}_s\Psi^{2,p-1}_r\\
  &=c_{p+1}(\{\Psi^{1,p-1}_{r-1},\Psi^{2,p-1}_s\}_c+\{\Psi^{2,p-1}_r,\Psi^{1,p-1}_{s-1}\}_c)
  +\Psi^{2,p-1}_s\Psi^{1,p-2}_{r-1}-\Psi^{2,p-1}_r\Psi^{1,p-2}_{s-1}\\
  &\quad  -\Psi^{1,p}_r\Psi^{2,p-1}_s+\Psi^{1,p}_s\Psi^{2,p-1}_r.
\end{align*}
Since $\{\Psi^{1,p-1}_{r-1},\Psi^{2,p-1}_{s}\}_c+\{\Psi^{2,p-1}_{r-1},\Psi^{1,p-1}_s\}_c=0$ (by (4)), we get
$\{\Psi^{1,p-1}_{r-1},\Psi^{2,p-1}_s\}_c=\{\Psi^{1,p-1}_s,\Psi^{2,p-1}_{r-1}\}_c.$
Similarly, we obtain $\{\Psi^{2,p-1}_r,\Psi^{1,p-1}_{s-1}\}_c=\{\Psi^{2,p-1}_{s-1},\Psi^{1,p-1}_r\}_c$.
Therefore, we have
\begin{align*}
\{\Psi^{1,p-1}_{r-1},\Psi^{2,p-1}_s\}_c+\{\Psi^{2,p-1}_r,\Psi^{1,p-1}_{s-1}\}_c
&=\{\Psi^{1,p-1}_s,\Psi^{2,p-1}_{r-1}\}_c+\{\Psi^{2,p-1}_{s-1},\Psi^{1,p-1}_r\}_c\\
&=\Psi^{1,p-1}_r\Psi^{2,p-1}_s-\Psi^{1,p-1}_s\Psi^{2,p-1}_r.
\end{align*}
Thus, we get $\{\Psi^{1,p+1}_r,\Psi^{2,p}_s\}_c+\{\Psi^{2,p}_r,\Psi^{1,p+1}_s\}_c=0$.
\item We have
\begin{align*}
LHS
&=\{\Psi^{1,p}_r,\Psi^{2,p+1}_{s-1}-c_{p+2}\Psi^{2,p}_{s-1}\}_c+\{\Psi^{2,p+1}_{r-1}-c_{p+2}\Psi^{2,p}_{r-1},\Psi^{1,p}_s\}_c\\
&=\{\Psi^{1,p}_r,\Psi^{2,p+1}_{s-1}\}_c+\{\Psi^{2,p+1}_{r-1},\Psi^{1,p}_s\}_c
-\Psi^{2,p}_{s-1}\Psi^{1,p-1}_r+\Psi^{2,p}_{r-1}\Psi^{1,p-1}_{s-1} -c_{p+2}(\{\Psi^{1,p}_r,\Psi^{2,p}_{s-1}\}_c+\{\Psi^{2,p}_{r-1},\Psi^{1,p}_s\}_c)\\
&=\{\Psi^{1,p}_r,\Psi^{2,p+1}_{s-1}\}_c+\{\Psi^{2,p+1}_{r-1},\Psi^{1,p}_s\}_c
-\Psi^{2,p}_{s-1}\Psi^{1,p-1}_r+\Psi^{2,p}_{r-1}\Psi^{1,p-1}_{s-1}
- c_{p+2}(\Psi^{1,p}_s\Psi^{2,p}_r-\Psi^{1,p}_r\Psi^{2,p}_s).
\end{align*}
Now we have
\begin{align*}
&\quad \  \{\Psi^{1,p}_r,\Psi^{2,p+1}_{s-1}\}_c+\{\Psi^{2,p+1}_{r-1},\Psi^{1,p}_s\}_c
\\
&=\{\Psi^{1,p}_r,\Psi^{0,p+1}_s-c_0\Psi^{1,p+1}_s\}_c+\{\Psi^{0,p+1}_r-c_0\Psi^{1,p+1}_r,\Psi^{1,p}_s\}_c\\
&=\{\Psi^{1,p}_r,\Psi^{0,p+1}_s\}_c+\{\Psi^{0,p+1}_r,\Psi^{1,p}_s\}_c-c_0(\{\Psi^{1,p}_r,\Psi^{1,p+1}_s\}_c+\Psi^{1,p+1}_r,\Psi^{1,p}_s)
-\Psi^{1,p+1}_s\{\Psi^{1,p}_r,c_0\}_c\\
&\quad -\Psi^{1,p+1}_r\{c_0,\Psi^{1,p}_s\}_c\\
&=\Psi^{1,p+1}_s\Psi^{2,p}_r-\Psi^{1,p+1}_r\Psi^{2,p}_s.
\end{align*}
Therefore, we get
\begin{align*}
LHS&=\Psi^{1,p+1}_s\Psi^{2,p}_r-\Psi^{1,p+1}_r\Psi^{2,p}_s- c_{p+2}(\Psi^{1,p}_s\Psi^{2,p}_r-\Psi^{1,p}_r\Psi^{2,p}_s)
 -\Psi^{2,p}_{s-1}\Psi^{1,p-1}_r+\Psi^{2,p}_{r-1}\Psi^{1,p-1}_{s-1}\\
&=\Psi^{2,p}_r\Psi^{1,p-1}_{s-1}-\Psi^{2,p}_s\Psi^{1,p-1}_{r-1}+\Psi^{2,p}_{r-1}\Psi^{1,p-1}_s-\Psi^{2,p}_{s-1}\Psi^{1,p-1}_r
=RHS.
\end{align*}
\end{enumerate}
\end{proof}


\begin{thebibliography}{10}

\bibitem{ABS}
V~Adler, A~Bobenko, and Y~Suris, \emph{Classification of integrable equations
  on quad-graphs. {T}he consistency approach}, Commun. Math. Phys. \textbf{3}
  (2003), no.~233, 513--543.

\bibitem{Bruschi1991}
M~Bruschi, O~Ragnisco, P~M Santini, and T~G Zhang, \emph{Integrable symplectic
  maps}, Phys. D \textbf{49} (1991), no.~3, 273--294.

\bibitem{CapelNijhoff1991}
H~W Capel, F~W Nijhoff, and V~G Papageorgiou, \emph{Complete integrability of
  Lagrangian mappings and lattices of KdV type}, Phys. Let. A \textbf{155} (1991),
  no.6-7, 377-387.

  \bibitem{DynaSys4}
  B. A. Dubrovin, I. M. Krichever, and S. P. Novikov, \emph{Integrable systems. I}, in  Dynamical systems, IV,
Encyclopaedia Math. Sci., vol. 4, Springer, Berlin, (2001),  177--332.

\bibitem{Grammaticos2004Integrability-o}
B~Grammaticos and A~Ramani, \emph{Integrability -- and how to detect it}, in
  ``Integrability of Nonlinear Systems'' (Y~Kosmann-Schwarzbach, B~Grammaticos,
  and K~M Tamizhmani, eds.), Lect. Notes Phys., vol. 638, Springer, 2004, 31--94.

\bibitem{Hirota}
R~Hirota, \emph{{Non-linear partial difference equations I-IV}}, J. Phys. Soc. Jpn. \textbf{43} (1977),
  1423--1433, 2074--2078, 2079--2086, J. Phys. Soc. Jpn.
  \textbf{45} (1978), 321--332.

\bibitem{Iatrou2003Higher-Dimension}
A~Iatrou \emph{Higher dimensional integrable mappings }, Physica D {\bf 179} (2003),  229--253

\bibitem{MaedaCompleteIntegrable}
S~Maeda, \emph{Completely integrable symplectic mapping}, Proc. Japan Acad.,
  Ser. A \textbf{6} (1983), 198--200

\bibitem{NCPquantummap}
F~W Nijhoff, H~W Capel, and V~G Papageorgiou, \emph{Integrable quantum
  mappings}, Phys. Rev. (3) \textbf{46} (1992), no. 4, 2155--2158 .


\bibitem{QCPN}
G~R~W Quispel, H~W Capel, V~G Capel, and F~W Nijhoff, \emph{Integrable mappings
  derived from soliton equations}, Phys. A \textbf{173} (1991), 243--266.

\bibitem{QCPN1984LinearIntegralEqs}
G~R~W Quispel, F~W Nijhoff, , H~W Capel, and J~van~der Linden, \emph{Linear
  integral equations and non-linear difference-difference equations}, Physica A
  \textbf{125} (1984), no. 2-3, 344--380.

\bibitem{RQs}
J~A~G Roberts and G~R~W Quispel, \emph{Constant Poisson structures for
difference equations, including integrable maps}, (in preparation).

\bibitem{Tranclosedform}
D~T Tran, P~H van~der Kamp, and G~R~W Quispel, \emph{Closed-form expressions
  for integrals of traveling wave reductions of integrable lattice equations},
  J.Phys A.: Math. Theor.  \textbf{42} (2009), no. 22, 225201 (20pp).


\bibitem{Tranfun}
D~T Tran, P~H van~der Kamp, and G~R~W Quispel, \emph{Functional independence
  of integrals of sine-Gordon, mKdV and pKdV maps},  (in preparation).

\bibitem{Kamp2009InitialValue}
P~H van der Kamp, \emph{Initial value problems for lattice equations}, J. Phys. A: Math. Theor. {\bf{42}}
(2009), 404019--34


\bibitem{KRQ}
P~H van~der Kamp, O~Rojas, and G~R~W Quispel, \emph{{Closed-form
  expressions for integrals of mKdV and sine-Gordon maps}}, J. Phys. A: Math.
  Theor. \textbf{39} (2007), 12789--12798.

\bibitem{Ves}
A~P Veselov, \emph{Integrable maps}, Russ. Math. Surveys \textbf{46} (1991),
  no.~5, 1--51.

\end{thebibliography}

\bibliographystyle{amsplain}

\end{document}